\documentclass[a4paper, 11pt]{article}

\usepackage{etex}
\usepackage[a4paper,top=1in,bottom=1in,left=1in,right=1in]{geometry}
\usepackage{caption}
\usepackage{booktabs}
\usepackage{fancyhdr}

\usepackage{scalefnt}
\usepackage{trfsigns}
\usepackage{sectsty}

\usepackage{appendix}
\usepackage{subfigure}

\usepackage{lipsum}

\usepackage{dsfont}

\usepackage[latin1]{inputenc}

\usepackage{rotating}
\usepackage{caption}
\usepackage{multirow}
\usepackage[longnamesfirst]{natbib}
\usepackage{amsmath, amssymb, amsthm, mathrsfs, enumitem}
\numberwithin{equation}{section}
\numberwithin{table}{section}
\numberwithin{figure}{section}
\usepackage{amsfonts}

\usepackage{lscape}
\usepackage{multirow} 
\usepackage{setspace}
\usepackage{rotating}
\usepackage{caption}
\usepackage{booktabs}
\usepackage{color}
\usepackage{ifpdf}

  \usepackage{graphicx}
  \DeclareGraphicsExtensions{.eps,.bmp,.png}
  \usepackage{pgf}
  \usepackage{tikz}
  \usepackage{pstricks}

\usepackage{epic}
\usepackage{floatflt}
\usepackage{wrapfig}
\newtheorem{theorem}{Theorem}[section]
\newtheorem{lemma}{Lemma}[section]

\theoremstyle{remark}
\newtheorem{assumption}{Assumption}

\newcommand \On {\mathscr{O}}

\DeclareMathAlphabet{\mathpzc}{OT1}{pzc}{m}{it}
\newcommand \Op {\mathscr{O}_{\mathpzc{p}}}
\newcommand \op {\mathpzc{o}_{\mathpzc{p}}}

\newcommand \ninf {n\rightarrow\infty}

\DeclareMathOperator{\plim}{plim}
\newcommand \ninv  {n^{-1}}

\definecolor{Grey}{RGB}{173,173,173}

\begin{document}

\title{Efficient closed-form estimation of large spatial autoregressions\thanks{I am grateful to the associate editor, three referees and Peter Robinson for very helpful comments that improved the paper significantly, and to Christos Kolympiris for permission to use the data in the empirical illustration.}}

\date{\today }
\author{Abhimanyu Gupta\thanks{Department of Economics, University of Essex, Wivenhoe Park, Colchester CO4 3SQ, UK. E-mail: a.gupta@essex.ac.uk.} \thanks{Research supported by ESRC grant ES/R006032/1.} }

\maketitle

\begin{abstract}
Newton-step approximations to pseudo maximum likelihood estimates of spatial autoregressive models with a large number of parameters are examined, in the sense that the parameter space grows slowly as a function of sample size. These have the same asymptotic efficiency properties as maximum likelihood under Gaussianity but are of closed form. Hence they are computationally simple and free from compactness assumptions, thereby avoiding two notorious pitfalls of implicitly defined estimates of large spatial autoregressions. When commencing from an initial least squares estimate, the Newton step can also lead to weaker regularity conditions for a central limit theorem than some extant in the literature. A simulation study demonstrates excellent finite sample gains from Newton iterations, especially in large multiparameter models for which grid search is costly. A small empirical illustration shows improvements in estimation precision with real data.
\end{abstract}

\textbf{Keywords:} Spatial autoregression, efficiency, many parameters, networks

\textbf{JEL Classification: }C21, C31, C33, C36

%
\clearpage
\allowdisplaybreaks
\section{Introduction}
Spatial autoregressive (SAR) models, introduced by \cite{cliff1973spatial}, are popular tools for modelling cross-sectionally dependent economic data. The pre-eminent feature of such models is the presence of one or more `spatial weight' matrices, which parsimoniously capture the dependence between units in the sample. Such dependence need not be geographic in nature, indeed the spatial weight matrix is known by other terms such as `adjacency matrix', `network link matrix' and `sociomatrix'. For $n\times 1$ vectors $y_n$ and $u$ of responses and unobserved disturbances, respectively, and an $n\times k$ covariate matrix $X_n$, the SAR model is 
\begin{equation}\label{hosar}
y_n=\sum_{i=1}^p\lambda_{0in}W_{in}y_n+X_n\beta_{0n}+u,
\end{equation}
where the elements of the $n\times n$ spatial weight matrices $W_{in}$ are inverse economic distances and $\lambda_{0n}=\left(\lambda_{01n},\ldots,\lambda_{0pn}\right)'$ and $\beta_{0n}$ are unknown parameter vectors. Subscripting with $n$ permits treatment of triangular arrays, an important issue for spatial models in general (see \cite{Robinson2011}), and for SAR models even more so due to various normalizations of the $W_{in}$ that make it $n$-dependent. This paper justifies computationally straightforward estimation for the parameters of (\ref{hosar}) with the same asymptotic properties as pseudo maximum likelihood estimates.

 SAR models allow dependence to occur across a very generalized notion of space: so long as a mapping exists between every pair of individuals to the real line a spatial weight matrix may be constructed. The flexible nature of the SAR model means that it may be used to model a very wide range of phenomena. Thus it has found application in many fields of economics such as development economics (\cite{case1991spatial}, \cite{Helmers2014}), industrial organization (\cite{pinkse2002}), trade (\cite{Conley2003}) and peer effects (\cite{Hsieh2018}), to name only a few examples. Another frequently used approach to model cross-sectional dependence is the `common factor' technique, see e.g. \cite{Chudik2015} for a review.

Estimation of SAR models has long been considered in the regional science literature, see e.g. \cite{anselin1988spatial}. Rigorous asymptotic theory for instrumental variables (IV) estimation was initially provided by \cite{kelejian1998generalized}, leading to the present flourishing theoretical literature. \cite{lee2002consistency} studied ordinary least squares (OLS) estimation of SAR models, stressing the need for lack of sparsity in the spatial weight matrix to establish desirable asymptotic properties such as consistency and efficiency. This was followed by \cite{lee2004asymptotic}, a seminal contribution that provided a taxonomical asymptotic theory for Gaussian pseudo maximum likelihood estimates (PMLE) of SAR models. Recently \cite{Kuersteiner2013,Kuersteiner2020} have provided general theory for such models in a panel data setting.

The flexible nature of SAR modelling is further embellished by the seamless ability to integrate more than one spatial weight matrix in the model (\ref{hosar}), thus permitting simultaneous connections between units across a number of channels. This is an accurate representation of typical economic situations, e.g. countries are `connected' by both geographical proximity as well as trade ties. Furthermore, in many economic settings the sample partitions naturally into $p$ clusters or groups, leading to block diagonal structure for the spatial weight matrix $V_n=diag\left(V_{1n},\ldots,V_{pn}\right)$, where $V_{in}$ is $m_i\times m_i$ and $\sum_{i=1}^p m_i=n$. To permit the modelling of heterogenous spillover effects across clusters, one may take $W_{in}$ to be the $n\times n$ block diagonal matrix with the $m_i\times m_i$ dimensional $i$-th diagonal block given by $V_{in}$. This approach has been suggested by \cite{Gupta2013,Gupta2018}. Here and more generally, the specification (\ref{hosar}) is termed a `higher-order' SAR model if $p>1$, see e.g \cite{blommestein1983specification}, \cite{Lee2010}, \cite{Li2017}, \cite{Han2017}, \cite{Kwok2019}. 

In the study of higher-order SAR models, \cite{Gupta2013,Gupta2018} have suggested that $p,k$ be allowed to diverge slowly to infinity as functions of sample size. The motivation for such generality is typically threefold: first, it is desirable to permit a richer model as the sample size permits. Second, clustered data as mentioned in the previous paragraph naturally imply asymptotic regimes with increasing $p$. For instance, when $m_i=m$ for each $i=1,\ldots,p$, we have $n=mp$ and the results of \cite{lee2004asymptotic} imply that $p\rightarrow\infty$ is necessary for consistent estimation, analogous to the problems created in the spatial statistics literature by `infill asymptotics', see e.g. \cite{Lahiri1996}. Finally, a theory that allows the model dimension to grow with sample size provides a more incisive analysis of large models in practice, much as typical asymptotic theory with a fixed parameter space itself can be thought as providing an approximation in finite samples. 

The estimation of such increasing-order SAR models has been studied by \cite{Gupta2013,Gupta2018} using IV, OLS and PMLE approaches. The first two methods have the advantage of being in closed-form, while even for $p=1$ PMLE (in)famously requires grid search and the inversion of an $n\times n$ matrix in every iteration, leading to many ingenious solutions for faster computation, see e.g.  \cite{Ord1975} and \cite{Pace1997}. The computational cost of PMLE in SAR type models is particularly salient as data sets increase in size, as stressed by \cite{Zhu2020}. Modern network data sets are amenable to modelling via SAR techniques and can feature, or accommodate, large parameter spaces but computation remains a serious challenge. \cite{Han2020} provide a discussion of the problems and propose a Bayesian solution.

These problems are naturally exacerbated if $p>1$, with grid search requiring more iterations to converge and each iteration requiring inversion of an $n\times n$ matrix, as well as risk of convergence to local optima. Furthermore, the requirement of a compact parameter space for $\lambda_{0n}$ can severely restrict the admissible parameter values (see \cite{Gupta2018}). On the other hand, under Gaussianity the PMLE becomes the MLE and is efficient. This property is shared by OLS, but under rather delicate and specific conditions even for $p=1$ (see \cite{lee2002consistency}). Thus the IV/OLS and PMLE approaches each have their advantages and it is desirable to combine the positive properties of both.

One method of obtaining closed-form estimates with the same asymptotic covariance matrix as a target estimate is to use Newton-type iterations commencing from an initial consistent estimator that is straightforward to compute. The approach dates back at least to \cite{Fisher1925} and \cite{LeCam1956}. It enjoys the added attraction of avoiding a potentially complicated consistency proof for an implicitly defined estimate, as well as the compactness assumptions this typically entails. As a result, the technique has been used in a vast variety of settings, see e.g. \cite{Rothenberg1964} (simultaneous equations), \cite{Hartley1965} (nonlinear least squares), \cite{Janssen1985} ($M$-estimation),  \cite{Rothenberg1984} (generalized least squares), \cite{Hualde2011}, \cite{Kristensen2006}, \cite{Robinson2005a} (time series and adaptive estimation),  \cite{Andrews1997} (generalized method of moments), \cite{Kasahara2008}, \cite{Kristensen2017} (structural estimation), \cite{DeLuca2018} (generalized linear models) and \cite{Frazier2017} (efficient two-step estimation), to name just a few.

In this paper we use IV and OLS estimates as initial estimates to form a single Newton-step asymptotic approximation to the Gaussian PMLE with $p=p_n$ and $k=k_n$ allowed to diverge as functions of $\ninf$. The approach has been studied in the case of fixed-dimensional SAR models by \cite{robinson2010efficient} and \cite{Lee2013}, but the previous discussion hints at its particular usefulness when considering large models. One avoids grid search over a high-dimensional parameter space, compactness assumptions on this space and the inversion of large ($n\times n$) matrix for every search iteration, as well as various headaches related to convergence and local optima. When commencing from IV estimates, this leads to closed-form efficient estimates under Gaussianity. As suggested by the results of \cite{lee2002consistency} and \cite{Gupta2013}, commencing iteration from OLS preserves the efficiency property. However, we show that the Newton step approach cancels out certain terms of large stochastic order that allows for weaker rate conditions than those imposed in these papers. 

In a simulation study, we demonstrate that the Newton step can lead to much improved estimates in finite samples, both in terms of bias and efficiency. While a single step is sufficient to establish desirable asymptotic properties, in our simulation study we also explore the finite sample implications of additional Newton steps, reporting results with up to six iterations. We find large finite sample gains in both bias and mean squared error that are robust to heavy tailed error distributions.  We also observe fast convergence of iterations, which conforms to extant theoretical observations. The gains are particularly notable when the parameter space and sample size is large, a situation in which PMLE becomes computationally onerous. In a small illustration with real world data, we show that the estimates work well in practice and lead to more precise results.  

We collect some frequently used notation here for the convenience of the reader.
For a generic matrix $A$ denote $\left\Vert A\right\Vert=\left(\overline{\eta}\left(A'A\right)\right)^{\frac{1}{2}}$, with $\overline\eta(\cdot)$ and $\underline\eta(\cdot)$ denoting the largest and smallest eigenvalues, respectively, of  a symmetric positive semidefinite matrix. Note that if $A$ is a vector then $\left\Vert A\right\Vert$ is simply its Euclidean norm. Let $\left\Vert A\right\Vert_R$ denote the maximum absolute row sum norm of $A$. For any parameter $\tau$, function $f(\tau)$ and generic estimate $\check\tau$, we will write $\check f\equiv f\left(\check\tau\right)$. We denote true parameter values with $0$ subscript and suppress the argument for a quantity evaluated at a true parameter value, i.e. $f\left(\tau_0\right)\equiv f$.
\section{Approximations to Gaussian PMLE}\label{sec:newton}
The ($-2/n$ times) log pseudo Gaussian likelihood function for model (\ref{hosar}) at any admissible point $\theta=(\lambda',\beta')'$ is given by
\begin{equation}\label{likelihood}
{\mathcal{Q}_n\left(\theta,\sigma^2\right)}=\log{(2\pi\sigma^2)}-\frac{2}{n}\log {\left\vert{S_{n}\left(\lambda\right)}\right\vert}
 +  \frac{1}{n\sigma^2}\left(S_n(\lambda)y_n-X_n\beta\right)'\left(S_n(\lambda)y_n-X_n\beta\right),
\end{equation} 
where $S_n(\lambda)=I_n-\sum_{i=1}^{p_n} \lambda_{in}W_{in}$, with $I_n$ denoting the $n\times n$ identity matrix.  If $S_n$ is invertible, (\ref{hosar}) admits the reduced form $y_n=S_n^{-1}X_n\beta_{n}+S_n^{-1}u$, and we  define $R_n=A_n+B_n,$ where $A_n = (G_{1n}X_n\beta_{n},\ldots,G_{{p_n}n}X_n\beta_{n}),B_n = (G_{1n}u,\ldots,G_{{p_n}n}u),$
$G_{in}(\lambda)=W_{in}S_n^{-1}(\lambda)$, $i=1,\ldots,p_n$, and so $R_n=(W_{1n}y_n,\ldots,W_{{p_n}n}y_n)$. 

Defining $\mathcal{R}_n^{y}\left(\theta\right)=R_n\lambda_n+X_n\beta_n-y_n$, the derivative of (\ref{likelihood}) at any admissible $\left(\theta,\sigma^2\right)$ is
\begin{equation}\label{xiany}
\xi_n\left(\theta,\sigma^2\right)=\left(\varphi_{n} {'}(\theta,\sigma^2),\;{2}{\sigma^{-2}}\ninv  \mathcal{R}_n^{y}{'}\left(\theta\right)X_n\right)',
\end{equation} 
where
$
 \varphi_n \left(\theta,\sigma^2\right)=2\sigma^{-2}\ninv \left(\sigma^{2}tr G_{1n}(\lambda)+y_n'W_{1n}'\mathcal{R}_n^{y}\left(\theta\right),\ldots,\sigma^{2}tr G_{p_nn}(\lambda)+y_n'W_{p_nn}'\mathcal{R}_n^{y}\left(\theta\right)\right)'.
$
Because $\mathcal{R}_n^{y}=-u$, denoting 
$
\phi_n   =  \sigma_0^{-2}\ninv \left( \sigma_0^2 tr C_{1n}-u'C_{1n}u ,\ldots,\sigma_0^2 tr C_{p_n}-u'C_{p_nn}u  \right)'$ with $C_{in}=G_{in}+G_{in}'$, we obtain
\begin{equation}\label{xidef}
\xi_n\equiv\frac{\partial \mathcal{Q}_n }{\partial\theta}=\left(\phi_n {'},0\right)'-{2}{\sigma_0^{-2}}t_n,
\end{equation}
with $t_{n}  = \ninv \left[A_n,X_n\right]'u$. The Hessian at any admissible point in the parameter space is

\begin{equation}\label{hessdef}
H_n\left(\theta_{n},\sigma^2\right)\equiv\frac{\partial^2\mathcal{Q}_n(\theta_{n},\sigma^2)}{\partial\theta\partial\theta'}=\left(\begin{array}{cc} \frac{2}{n} P_{ji,n}(\lambda_{n})+\frac{2}{n\sigma^2}R_n'R_n & \frac{2}{n\sigma^2}R_n'X_n \\ \\ \frac{2}{n\sigma^2}X_n'R_n & \frac{2}{n\sigma^2}X_n'X_n \end{array} \right)
\end{equation}
where $P_{ji,n}(\lambda_{n})$ is the $p_n\times p_n$ matrix with $(i,j)$-th element given by \linebreak
$tr \left(G_{jn}(\lambda_{n})G_{in}(\lambda_{n})\right)$.


 Let $Z_n$ be an $n\times r_n$ matrix of instruments, with $r_n\geq p_n$, and  define the IV and OLS estimates as 
\begin{eqnarray}
{\hat{\theta}}_{n} & = & \hat{Q}_n^{-1}\hat{K}_n'J_n^{-1}\hat{k}_n,\;\;\;\;\;\;\;\;{\hat{\sigma}}^2_{n}=\ninv \left\Vert y_n-\left(R_n,X_n\right)\hat{\theta}_{n} \right\Vert^2, \label{iv} \\
{\tilde{\theta}}_{n} & = & \hat{L}^{-1}_n\hat{l}_n ,\;\;\;\;\;\;\;\;\;\;\;\;\;\;\;\;\;\;\;\;{\tilde{\sigma}}^2_{n}=\ninv \left\Vert y_n-\left(R_n,X_n\right)\tilde{\theta}_{n} \right\Vert^2, \label{ols}
\end{eqnarray} 
respectively, with
$\hat{Q}_n=\hat{K}_n'J_n^{-1}\hat{K}_n$, 
$
\hat{K}_n=\ninv \left[Z_n, X_n\right]'[R_n,X_n],\hat{k}_n=\ninv \left[Z_n , X_n\right]'y_n,\;J_n=\ninv \left[Z_n ,X_n\right]'\left[Z_n,X_n\right],
$ and 
$
\hat{L}_n=\ninv \left[R_n, X_n\right]'[R_n,X_n],\hat{l}_n=\ninv \left[R_n,X_n\right]'y_n.
$
Define the respective `one-step' estimates $\hat{\hat{\theta}}_{n}$ and $\tilde{\tilde{\theta}}_{n}$ by the following equations
\begin{eqnarray}
\hat{\hat{\theta}}_{n} & = & \hat{\theta}_{n}-\hat{H}_n^{-1}\hat{\xi}_n, \label{ivstep} \\
\tilde{\tilde{\theta}}_{n} & = & \tilde{\theta}_{n}-\tilde{H}_n^{-1}\tilde{\xi}_n. \label{olsstep}
\end{eqnarray} 
We observe that other initial estimates, such as the GMM estimates of \cite{Kelejian1999} and \cite{Lee2007}, can also be used. However we choose initial estimates that are available in closed form for computational ease. While consistent initial estimates are needed to obtain a desirable asymptotic theory, even in the fixed-dimension parametric case these are permitted to be $n^\psi$-consistent, where $\psi<1/2$, see \cite{Robinson1988} and references therein.

While our theorems below establish desired asymptotic properties for the one step estimates, from a practical point of view more iterations may be desirable. In fact, these also improve the statistical rate of convergence to the target PMLE, yielding an even faster statistical counterpart to the famous quadratic numerical rate of convergence of Newton estimates, see for example Theorem 2 of \cite{Robinson1988} and p. 312-313 of \cite{Ortega1970}. We examine this issue in more detail in the next section and also the Monte Carlo study. 
\section{Asymptotic properties}\label{sec:asymptotics}
The following assumptions are discussed in \cite{lee2002consistency,lee2004asymptotic}, and \cite{Gupta2013,Gupta2018}, amongst other spatial papers in which they are routinely employed. These conditions are by no means the weakest possible set, but we opt for tractability to convey the main message especially in view of the large number of spatial parameters involved. For example, stochastic regressors can be easily accommodated but complicate the notation.
\begin{assumption}\label{errors}
$u=(u_{1},\ldots,u_{n})'$ has iid elements with zero mean and finite variance $\sigma_0^2$.
\end{assumption}
\begin{assumption}\label{weights}
For $i=1,\ldots,p_n$, the elements of $W_{in}$ are uniformly $\On\left(1/h_n\right)$, where $h_n$ is some positive sequence which may be bounded or divergent, but always bounded away from zero and such that $n/h_n\rightarrow\infty$ as $\ninf$. The diagonal elements of each $W_{in}$ are zero. 
\end{assumption}
\begin{assumption}\label{sn}
$S_n$ is non-singular for all sufficiently large $n$.
\end{assumption}

\begin{assumption}\label{rcsums}
$\left\Vert S_n^{-1} \right \Vert_R$, $\left\Vert S_n'^{-1} \right \Vert_R$, $\left\Vert W_{in} \right \Vert_R$ and $\left\Vert W'_{in} \right \Vert_R$ are uniformly bounded in $n$ and $i$ for all $i=1,\ldots,p_n$ and sufficiently large $n$.
\end{assumption}
\begin{assumption}\label{x} 
The elements of $X_n$ are constants and are uniformly bounded in $n$, in absolute value, for all sufficiently large $n$.
\end{assumption}
\begin{assumption}\label{z}
The elements of $Z_n$ are constants and are uniformly bounded in absolute value, for all sufficiently large $n$.
\end{assumption}
\begin{assumption}\label{K}
$\displaystyle\varlimsup_{n\rightarrow\infty}\overline{\eta}(J_n)<\infty$ and $\displaystyle\varliminf_{n\rightarrow\infty}\underline{\eta}(K_n'K_n)>0$.
\end{assumption}
\begin{assumption}\label{J}
$\displaystyle\varliminf_{n\rightarrow\infty}\underline{\eta}(J_n)>0$ and $\displaystyle\varlimsup_{n\rightarrow\infty}\overline{\eta}(K_n'K_n)<\infty$.
\end{assumption}
\begin{assumption}\label{l2}
$\displaystyle\varlimsup_{n\rightarrow\infty}\overline{\eta}(L_n)<\infty$.
\end{assumption}
\begin{assumption}\label{l1}
$\displaystyle\varliminf_{n\rightarrow\infty}\underline{\eta}(L_n)>0$.
\end{assumption}
\begin{assumption}\label{errors2}
$\mathbb{E}\left( u_{i}^{4}\right)\leq C$ for $i=1,\ldots,n$.
\end{assumption}
Let $\Psi_n$ be an $s{\times}(p_n+k_n)$ matrix of constants with full row-rank.
The claims of the following theorems also hold when $p_n$ and $k_n$ are fixed, but we state and prove the results for the more challenging case when these diverge. 
\begin{theorem}\label{approx}
\mbox{}
\begin{itemize}
\item[(i)] Let Assumptions \ref{errors}-\ref{errors2} hold along with
\begin{equation}\label{approx2sls1}
\frac{1}{p_n}+\frac{1}{r_n}+\frac{1}{k_n}+\frac{p_n^3k_n^4}{n}+\frac{p_n^{\frac{3}{2}}k_n}{h_n}\rightarrow 0  \text{ as } \ninf,
\end{equation}
and
\begin{equation}\label{approx2sls4}
\frac{r_n^2}{n} \text{ bounded as } \ninf.
\end{equation}
Then
\[
\frac{n^{\frac{1}{2}}}{\left(p_n+k_n\right)^{\frac{1}{2}}}\Psi_n\left(\hat{\hat{\theta}}_{n}-\theta_{0n}\right)\overset{d}\longrightarrow N\left(0,\lim_{\ninf}\frac{\sigma_0^2}{p_n+k_n}\Psi_n L_n^{-1}\Psi_n'\right),
\]
where the asymptotic covariance matrix exists, and is positive definite, by Assumptions \ref{l2} and \ref{l1}.
\item[(ii)] Let Assumptions \ref{errors}-\ref{x} and \ref{l2}-\ref{errors2} hold. Suppose also that
\begin{equation}\label{approxols1}
\frac{1}{p_n}+\frac{1}{k_n}+\frac{p_n^3k_n^4}{n}+\frac{p_n^{\frac{3}{2}}k_n}{h_n}+\frac{n^{\frac{1}{2}}p_n^{\frac{5}{2}}}{h_n^3}\rightarrow 0. 
\end{equation}

Then
\[
\frac{n^{\frac{1}{2}}}{\left(p_n+k_n\right)^{\frac{1}{2}}}\Psi_n\left(\tilde{\tilde{\theta}}_{n}-\theta_{0n}\right)\overset{d}\longrightarrow N\left(0,\lim_{\ninf}\frac{\sigma_0^2}{p_n+k_n}\Psi_n L_n^{-1}\Psi_n'\right),
\]
where the asymptotic covariance matrix exists, and is positive definite, by Assumptions \ref{l2} and \ref{l1}.
\end{itemize}

\end{theorem}

In the `just identified' case $p_n=r_n$, condition (\ref{approx2sls4}) is implied by (\ref{approx2sls1}). Theorem \ref{approx} $(i)$ shows that the one-step estimate asymptotically achieves the efficiency bound noted by \cite{lee2002consistency}. On the other hand, Theorem \ref{approx} $(ii)$ yields the same distributional result as for the OLS estimate (Theorem 4.3 of \cite{Gupta2013}). This should come as no surprise since \cite{lee2002consistency} has already established the efficiency of OLS under suitable conditions. Nevertheless, Theorem \ref{approx} $(ii)$ imposes weaker conditions on the relative rates of $h_n$ and $n^{\frac{1}{2}}$ than those extant in the literature. 

Indeed, for their result, \cite{Gupta2013} assumed ${n^{\frac{1}{2}}p_n^{\frac{1}{2}}}/{h_n}\rightarrow 0$ as compared to our ${n^{\frac{1}{2}}p_n^{\frac{5}{2}}}/{h_n^3}\rightarrow 0$. The latter is a quantity of smaller order as $\left({n^{\frac{1}{2}}p_n^{\frac{5}{2}}}/{h_n^3}\right)/\left({n^{\frac{1}{2}}p_n^{\frac{1}{2}}}/{h_n}\right)=(p_n/h_n)^2\rightarrow 0$. For fixed $p_n$ and $k_n$, our asymptotic normality result relies only on $n^{\frac{1}{2}}/h_n^{3}\rightarrow 0$, as $\ninf$. This is a weaker requirement as compared to \cite{lee2002consistency}, who assumed $n^{\frac{1}{2}}/h_n\rightarrow 0$ as $\ninf$. The reason for these favourable outcomes is the cancellation of higher order terms when using the one-step approximation. The key difference is in the rates $\left\Arrowvert n^{-1}\left[B_n, 0\right]'u\right\Arrowvert= \Op\left({p_n^{\frac{1}{2}}}/{h_n}\right)$
and $\left\Arrowvert\phi_n\right\Arrowvert=\Op\left({p_n^{\frac{1}{2}}}/{n^{\frac{1}{2}}h_n^{\frac{1}{2}}}\right),$
the latter being sharper since $n/h_n\rightarrow\infty$ as $\ninf$. 

To more transparently illustrate the implications of our weaker rate conditions, consider data collected in a `farmer-district' type of environment, such as in \cite{case1991spatial}. Suppose that there are $D$ districts, each containing $m$ farmers, so that $n=Dm$, and $D,m\rightarrow\infty$ simultaneously. There is independence across districts, but equal dependence within districts, yielding $h_n=m-1$ (see \cite{lee2002consistency,lee2004asymptotic} for a more detailed discussion). Then, with fixed $p_n$, \cite{lee2002consistency} and \cite{Gupta2013} required $n^{\frac{1}{2}}/h_n=D^{\frac{1}{2}}/m^{\frac{1}{2}}=o(1)$, while our condition imposes $n^{\frac{1}{2}}/h_n^3=D^{\frac{1}{2}}/m^{\frac{5}{2}}=o(1)$. Thus, our condition permits $D$ to grow much faster as we only need $D^{\frac{1}{5}}=o(m)$ as compared to $D=o(m)$. The author thanks an anonymous referee for suggesting this illustration. We note that \cite{robinson2010efficient} obtained asymptotic normality, indeed efficiency, in a semiparametric setup with $p_n=1$ requiring only $h_n\rightarrow\infty$ if the disturbances are symmetrically distributed or the weight matrix is symmetric. This condition would likely need to be suitably amended as $p_n\rightarrow\infty$.

If $h_n$ is bounded as $\ninf$, a more complicated analysis is required to establish that one-step estimates achieve the PMLE asymptotic covariance matrix, because the information equality does not hold asymptotically. Denote $\mu_l=\mathbb{E}\left(u_i^l\right)$ for natural numbers $l$, and introduce, with $i,j=1,\ldots,p_n$, the $p_n \times p_n$ matrix $\Omega_{\lambda\lambda,n}$ with $(i,j)$-th element
$
\frac{4\mu_3}{n\sigma_0^4}\sum_{r=1}^n c_{rr,in}b_{r,jn}X_n\beta_{0n}+\frac{\left(\mu_4-3\sigma_0^4\right)}{n\sigma_0^4}\sum_{r=1}^n c_{rr,in}c_{rr,jn}
$
and the $k_n \times p_n$ matrix $\Omega_{\lambda\beta,n}$ with $i$-th column
$
\frac{2\mu_3}{n\sigma_0^4}\sum_{r=1}^n c_{rr,in}x_{r,n}
$
where $c_{pq,in}$ is the $(p,q)$-th element of $C_{in}$, $b_{jn}=G_{jn}X_n\beta_{0n}$ with $t$-th element $b_{t,jn}$ ($j=1,\ldots,p_n$ and $t=1,\ldots,n$) and $x_{p,n}$ is the $p$-th column of $X_n'$. Define
\begin{equation}\label{omegadef}
\Omega_n=\left(\begin{array}{cc} \Omega_{\lambda\lambda,n} & \Omega_{\lambda\beta,n}' \\ \\ \Omega_{\lambda\beta,n} & 0 \end{array} \right).
\end{equation}
Then 
$
\mathbb{E}\left(\xi_n\xi_n'\right)=\ninv \left(2\Xi_n+\Omega_n\right),
$
where 
\begin{equation}\label{Xidef}
\Xi_n=\mathbb{E}\left(H_n\right)=\left(\begin{array}{cc} \frac{2}{n} \left(P_{ji,n}+P_{j'i,n}+\frac{1}{\sigma_0^2}A_n'A_n\right) & \frac{2}{n\sigma_0^2}A_n'X_n \\ \\ \frac{2}{n\sigma_0^2}X_n'A_n & \frac{2}{n\sigma_0^2}X_n'X_n \end{array} \right).
\end{equation}
When $h_n$ is bounded OLS cannot be consistent (see \cite{lee2002consistency}), so the following theorem considers only initial IV estimates.
\begin{theorem}\label{approxbdd}
Let Assumptions \ref{errors}-\ref{K} hold. Suppose that $h_n$ is bounded away from zero and that there is a real number $\delta>0$ such that
$
\mathbb{E}\left\arrowvert u_{i}\right\arrowvert^{4+\delta}\leq C 
$
for $i=1,\ldots,n$. In addition, assume that
\begin{equation}\label{hbddcovass}
\varlimsup_{\ninf}\overline{\eta}\left(2\Xi_n^{-1}+\Xi_n^{-1}\Omega_n\Xi_n^{-1}\right)<\infty,\;\varliminf_{\ninf}\underline{\eta}\left(2\Xi_n^{-1}+\Xi_n^{-1}\Omega_n\Xi_n^{-1}\right)>0\;\mathrm{and}\;\varliminf_{\ninf}\underline{\eta}\left(\Xi_n\right)>0.
\end{equation}
Suppose also that 
\begin{equation}\label{approx2sls3}
\frac{1}{p_n}+\frac{1}{r_n}+\frac{1}{k_n}+\frac{p_n^3k_n^2\left(p_n\left(r_n+k_n\right)+k_n^2\right)}{n}+\frac{\left(p_nk_n\right)^{2+\frac{8}{\delta}}}{n}\rightarrow 0 \text{ as } \ninf.
\end{equation}
Then
\[
\frac{n^{\frac{1}{2}}}{\left(p_n+k_n\right)^{\frac{1}{2}}}\Psi_n\left(\hat{\hat{\theta}}_{n}-\theta_{0n}\right)\overset{d}\longrightarrow N\left(0,\lim_{\ninf}\frac{\sigma_0^2}{p_n+k_n}\Psi_n \left(2\Xi_n^{-1}+\Xi_n^{-1}\Omega_n\Xi_n^{-1}\right)\Psi_n'\right),
\]
where the asymptotic covariance matrix exists, and is positive definite, by (\ref{hbddcovass}).
\end{theorem}

The rate condition (\ref{approx2sls3}) can simplify depending on the value of $\delta$, i.e. the order of the finite moments assumed for $u_i$. As $\delta$ grows larger, the last term in the rate condition becomes redundant, indeed the numerator therein tends to $p_n^2k_n^2$ as $\delta\rightarrow\infty$, which is evidently dominated by the numerator of the other rate restriction. In the `farmer-district' setting discussed earlier, we have bounded $h_n=m-1$ in this case. To further illustrate the rate condition, suppose that we are in the just identified case $p_n=r_n$. Then (\ref{approx2sls3}) requires $p_n^5k_n^3+p_n^4k_n^4+p_n^3k_n^5+\left(p_nk_n\right)^{2+\frac{8}{\delta}}=o(n)$. Then the term involving $\delta$ dominates the other three if $\delta\leq 8/3$. 

As indicated earlier, further iterations on the Newton step can improve the rate of statistical convergence to the target as well as finite sample properties. To see this, let $\hat{\hat\theta}^{\ell}_n$ be the $\ell$-th Newton iteration towards the PMLE $\check\theta_n$.
By  Theorem 2 of \cite{Robinson1988}, $\left\Vert\check\theta_n- \hat{\hat\theta}^{\ell+1}_n\right\Vert=\Op\left(\left\Vert\check\theta_n- \hat{\hat\theta}_n\right\Vert^{2^\ell}\right)$, an identical bound holding also for $\tilde{\tilde\theta}^{\ell+1}_n$. A factor that depends on $\ell$ is suppressed in the stated stochastic bound, indicating that this is not uniform in $\ell$. Because the results of \cite{Gupta2018} and this paper show that one-step Newton estimates and $\check\theta_n$ are $n^{1/2}/\left(p_n+k_n\right)^{1/2}$-consistent, we have
\[
\left\Vert\check\theta_n- \hat{\hat\theta}^{\ell+1}_n\right\Vert=\Op\left(\left(n/\left(p_n+k_n\right)\right)^{-2^{\ell-1}}\right),\;\;\;\;\;\left\Vert\check\theta_n- \tilde{\tilde\theta}^{\ell+1}_n\right\Vert=\Op\left(\left(n/\left(p_n+k_n\right)\right)^{-2^{\ell-1}}\right),
\] thus yielding the rate at which the iterations approximate the target estimate in a statistical sense, pointwise in $\ell$.


\section{Finite-sample performance of Newton-step estimates}\label{sec:mc}
\subsection{Fixed number of neighbours (bounded $h_n$)}\label{subsec:bddh}We examine finite-sample performance of $\hat{\hat{\theta}}_n$ in this section, since the IV case entails a change in limiting distribution due to the Newton step and OLS requires divergent $h_n$ to be consistent. Following \cite{das2003} and the design in \cite{Gupta2013}, define $W^*_{in}$ as the symmetric circulant matrix with first row
\begin{equation}\label{circulant1}
w^*_{1j,in}=\left\{\begin{array}{ll}
0 & \text{ if } j=1 \text{ or } j=i+2,\ldots,n-i ;\\
1 & \text{ if } j=2,\ldots,i+1 \text{ or } j=n-i+1,\ldots,n,\\
\end{array}\right.
\end{equation} 
and take $W^c_{in}={\left\Vert W^*_{in}\right\Vert^{-1}} W^*_{in},$
where $\left\Vert W^*_{in}\right\Vert=\overline{\eta}\left(W^*_{in}\right)=2i,$ because $W^*_{in}$ is a symmetric, circulant matrix (see e.g. \cite{Davis1979} p. 73). Thus $W^c_{in}$ is also a symmetric circulant matrix with first row given by $w^*_{1j,in}/2i$. This is an example of spatial weight matrices with bounded $h_n$. 

We now dispense with some $n$ subscripts for brevity. Our design generates $y=S^{-1}(X\beta+u)$ for sample sizes $n=200,400,800$ and $k=2$, with elements $x_{j1}$ and $x_{j2}$ of $X$ generated as iid replicates from a $U(0,1)$ distribution, $j=1,\ldots,n$. We generate the disturbance $u$ using two different distributions: $N(0,1)$ and $t_6$. PMLE becomes MLE under the first, while the second has heavier tails. Our experiments take $p=2,4,6$ for each of the described designs. We use a design with weights matrices given by $W^c_{i}$, $i=1,\ldots,p$. Finally, we set $\beta_1=1$, $\beta_2=0.5$ and $p=2:\lambda_1=0.4,\lambda_2=0.5;p=4: \lambda_1=0.3,\lambda_i=0.2,i=2,3,4;p=6: \lambda_i=0.15,i=1,\ldots,6.$ The choices of $\lambda_i$ satisfy the sufficient condition $\sum_{i=1}^p\left\vert\lambda_i\right\vert<1$ for invertibility of $S$.

With the aim of comparing initial IV estimates and MLE to Newton-step estimates, we first report three statistics: Monte Carlo  mean, Monte Carlo  mean squared error, and relative root Monte Carlo mean squared error, the latter being a straightforward ratio of the root MSE for IV and the iterated estimate. We also examine the use of more than one iteration in finite samples, and for this recall the notation $\hat{\hat\theta}^{\ell}_n$ for the $\ell$-th Newton iteration. Our results are reported for $\ell=1,3,6$. The set of instruments that we use for our initial estimates are the linearly independent columns of $Z=\left(W_{1}^cX,\ldots,W_p^cX,X\right)$.

In Tables \ref{table:norm_mean} and \ref{table:t6_mean}, we report the Monte Carlo mean of our estimates for standard normal and $t_6$ errors, respectively. For standard normal errors, we notice that the initial IV estimate can be heavily biased but Newton iterations improve matters, sometimes spectacularly. Indeed, for $p=6$ and $n=200$ the performance of $\hat\theta_n$ can be appalling, with $\hat\lambda_5<0$. However after six Newton steps this has improved to 0.1216 and even three iterations lead to a significant improvement. The reduction of bias from Newton iterations is not a universal feature, however broadly speaking the Newton steps reduce bias in the estimates, even for smaller values of $p$. As the sample size increases the iterations converge substantially, with little to choose typically between $\hat{\hat\theta}^3_n$ and $\hat{\hat\theta}^6_n$ for $n=800$. However for $n<800$, we notice that three iterations usually do the job quite satisfactorily, especially when $p<6$.

For $t_6$ errors, Table \ref{table:t6_mean} paints a similar picture to Table \ref{table:norm_mean}. Once again, the noticeable `rogue' estimate is for $\lambda_5$ when $p=6$ and $n=200$. Considering that all our simulations start from the same seed, this outlier may possibly be attributed to a bad draw. As in the normal errors case, results are quite stable for larger $n$ and smaller $p$, and typically show bias reduction due to Newton steps and near convergence after three iterations. 

Tables \ref{table:norm_rmse} and \ref{table:t6_rmse} report mean squared error (MSE) for the IV estimates and iterated estimates with $N(0,1)$ and $t_6$ errors, respectively. As may be expected, MSE is very high for designs that combine the largest values of $p$ with the smallest values of $n$. The efficiency improvement due to the Newton step is apparent, with iterations leading to very clear improvements (i.e. reductions) in MSE. These gains can be spectacular in many cases, for example for the $\lambda_i$ estimates when $p=6$ and $n=800$. These patterns of improvement with iteration are similar for both error distributions but the magnitude of MSE is generally much larger for $t_6$ errors, which features heavier tails than the normal distribution.

In Tables \ref{table:normal_rrmse} and \ref{table:t6_rrmse}, we report the ratio of the Monte Carlo root mean squared error of $\hat\theta_n$ to that of $\hat{\hat\theta}^\ell_n$, $\ell=1,3,6$, abbreviating this quantity to RRMSE. An RRMSE of two indicates that the RMSE of the IV estimate is twice that of the Newton iteration it is being compared to. Our results in Table \ref{table:normal_rrmse} show that Newton iterations can lead to tremendous finite sample gains in MSE. These gains are present in 100\% of the cases considered, but are generally larger for the spatial parameters $\lambda_i$ than the regression parameters $\beta_i$. 

We discuss the spatial parameter estimates first. Note that for greater sample sizes we have greater MSE gains, often the gains more than doubling from $n=200$ to $n=800$, and sometimes even tripling. As observed for the means in Table \ref{table:norm_mean}, there is usually not much to choose from between the third and sixth iterations. With and $n=800$ we nearly always obtain Newton estimates with RMSE a quarter of that for IV, and occasionally even a fifth of the IV RMSE. In most cases three iterations are enough to achieve these superb gains.

These patterns for the $\lambda_i$ qualitatively repeat themselves when the errors are $t_6$, as seen in Table \ref{table:t6_rrmse}. In this case when $n=800$ we achieve RMSE improvements over IV of a factor of 2.15 always when three iterations are carried out, with factors of three commonly seen and one case with nearly a fourfold improvement. The factors of efficiency improvement that we observe in our results can dominate similar precedents in other settings. Indeed, the greatest relative root MSE improvement that \cite{Robinson2005a}  finds in his fractional time series setting is $\sqrt{1/0.23}=2.085$ (see Table 4 of that paper).  

Moving to the estimates of the regression parameters $\beta_1$ and $\beta_2$, in both Tables \ref{table:normal_rrmse} and \ref{table:t6_rrmse} we see almost universal improvement over IV. The exceptions are four cases out of a total of 54 in Table \ref{table:t6_rrmse}, for the $t_6$ case. These RMSE gains are not as spectacular as for the $\lambda_i$, but are generally noticeably large as both $n$ and $p$ increase. Indeed, for $n=800$ we observe that the RMSE for the IV estimate can sometimes be almost one and a half times are large as the Newton iterations when $p=6$ and $n=800$. For $n\geq 400$, IV performs worse than the Newton iterations almost uniformly (there are only two exceptions for $t_6$ errors) over both $\beta_1$ and $\beta_2$, the values of $p$, the number of iterations and the error distribution. Thus there is evidence of the usefulness of Newton iterations even for the regression parameters, albeit the gains are greater for the spatial parameters.

Finally, we also present the RRMSE of MLE (denoted $\mathring\theta_n$) to our proposed iterated estimates in Table \ref{table:normal_mlerrmse}  for $N(0,1)$ errors. Naturally, we anticipate MLE to outperform iterated IV estimates for smaller sample sizes and, because our iterations target the MLE limiting covariance matrix, a reasonable aim is to approach the RMSE of the MLE as $n$ grows larger. Indeed, we find that this is the case. Recall that our estimates are designed to approximate but not outperform MLE: the main focus of the paper is computational simplicity. Our estimates are available in closed form and can be computed much faster than those requiring grid search and inversion of an $n\times n$ matrix. Thus, approaching the MLE in RMSE as $n$ grows is an encouraging and desirable property of our estimates. Finally, we observe that the RMSE of $\mathring\beta_n$ is much closer to the iterated estimates than is the case for $\mathring\lambda_n$. For the latter, larger sample sizes are needed for the RRMSE to approach unity.
\subsection{Growing number of neighbours (divergent $h_n$)}\label{subsec:divh}
In this section we explore the performance of the Newton step estimates when the number of neighbours diverges with sample size, i.e. $h_n\rightarrow\infty$. This design, with diverging $h_n$, also allows us to study the performance of iterations on OLS starting values. For each $i=1,\ldots, p$, we generate a $n\times n$ matrix $W_{in}^*$ as $w^*_{rs,in}=\Phi \left(-d_{rs,i}\right) I\left(c_{rs,i}<n^{1/3}/100\right)$ if $r\neq s$, and $w^*_{rr,in}=0$, where $\Phi (\cdot )$ is the standard
normal cdf, $d_{rs,i}\sim$iid $U[-3,3]$, and $c_{rs,i}\sim$iid $U[0,1]$. This construction generates $W_{in}^*$ with approximately $n^{1/3}\%$ (up to closest integer) nonzero elements. These $W_{in}^*$ are then symmetrized and normalized by spectral norm to ensure stability, yielding the final set of $W_{in}$ that we employ. The remaining design details are as in the previous subsection. To conserve space, we report results only for $N(0,1)$ errors. 
 
 Tables \ref{table:hdiviv_norm_mean} and \ref{table:hdivols_norm_mean} display the Monte Carlo mean of $\hat\theta_n$, $\hat{\hat\theta}^{1}_n$,  $\hat{\hat\theta}^{3}_n$, $\tilde\theta_n$, $\tilde{\tilde\theta}^{1}_n$ and $\tilde{\tilde\theta}^{3}_n$. Convergence of iterations is achieved after three Newton steps, so we do not report the sixth iteration as in the previous subsection. In fact, convergence is practically fully achieved by just a single iteration with the IV starting values $\hat\theta_n$, as Table \ref{table:hdiviv_norm_mean} indicates. Examining Table \ref{table:hdivols_norm_mean} suggests that a third iteration has  more influence for OLS starting values, but modestly so. Tables \ref{table:hdiviv_norm_rmse} and \ref{table:hdivols_norm_rmse} report MSE for the same sets of estimates and we find a similar pattern: for IV starting values one iteration seems to do the job and reduces MSE. On the other hand, for OLS starting values the first iteration increases MSE but the third iteration reduces it, following which performance is stable and so we do not report further iterations.

In Table \ref{table:hdiv_norm_rrmse} we report RRMSE of the estimates studies above. We notice that IV estimates improve in MSE with a single Newton step, and subsequent iterations do not help much, because convergence is achieved. On the other hand, when starting with OLS values $\tilde\theta_n$, further iterations are beneficial and yield more efficient estimates. Convergence is completely achieved after three iterations in this case. We also find that Newton steps, whether they commence from $\hat\theta_n$ or $\tilde\theta_n$, give greater efficiency gains for the spatial parameters $\lambda_i$ rather than the regression coefficients $\beta_i$. This matches the results in the previous subsection. Because the $\lambda_i$ correspond to the potentially endogenous spatial lags $W_{in}y$, we might expect initial estimates of these to have greater potential for improvement compared to the $\beta_i$.
 \subsection{Heteroskedastic errors}\label{subsec:het}
In this design, we confirm the robustness of our findings to heteroskedasticity in the error distribution. We generate the errors using multiplicative heteroskedasticity via the regressors, and report only the bounded $h_n$ weight matrices of Section \ref{subsec:bddh} and designs with Gaussian errors to conserve space. Specifically, we employ a $N(0,h_{jn})$ distribution for the errors, where $h_{jn}=n\left(\sum_{r=1}^n\left(\left\vert x_{r1}\right\vert+\left\vert x_{r2}\right\vert\right)\right)^{-1}\left(\left\vert x_{j1}\right\vert+\left\vert x_{j2}\right\vert\right)$, see \cite{Liu2015} and also \cite{Lin2010}. Monte Carlo mean, mean squared error and RRMSE of IV estimates to iterated Newton step estimates are presented in Tables \ref{table:hetnorm_mean}-\ref{table:hetnormal_rrmse}. We find the same qualitative patterns as were observed for the homoskedastic designs presented earlier, with the `rogue' IV estimate for $\lambda_5$ appearing again because we start our simulations from the same seed. As far as quantitative results are concerned, the improvements due to the Newton step are generally smaller than the homoskedastic case but still substantial.

\section{Empirical illustration}\label{sec:empirical}
In this small empirical illustration we show that the Newton step estimates perform well in practice and can lead to more precise estimation. The example is based on \cite{kolympiris2011spatial} (KKM), and is also studied in \cite{Gupta2013}. KKM seek to model the venture capital funding (provided by venture capital firms (VCFs)) for dedicated biotechnology firms (DBFs) with a SAR model. The hypothesis is that the level of VC funding for a DBF increases with the number of VCFs located in close proximity. Denoting by $d_{lk}$ the distance in miles between the $l$-th and $k$-th DBFs, we estimate
\begin{equation}\label{KKMspec}
y=\sum_{i=1}^{p}\lambda_i W_{i}^{b}y+X\beta+U,
\end{equation}
 where $W_{i}^b$ is the (row-normalised) weight matrix having off-diagonal $(l,k)$-th element equal to 1 if $i-1< d_{lk}\leq i$, $i=1\ldots,p,$ and if $d_{lk}=0$ for $i=1$. Thus the matrices are based on each one of $p$ sequential 1-mile rings from the origin DBF. $y$ is the vector of natural logs of the amount of VC funding (million \$) received by each of $n=816$ DBFs.

We first focus on estimates of the main parameters of interest $\lambda_i$ in (\ref{KKMspec}). We estimate (\ref{KKMspec}) with $p=2,4,6$ using initial IV and the Newton-step estimates that we have justified theoretically. We only report the Newton-step for a single iteration as convergence is achieved. Like \cite{Gupta2013}, we find that only $\lambda_1$ and $\lambda_2$ are statistically significant at the 1\% level, and the magnitude of our parameter estimates is also close to their findings, with our results reported in Table \ref{table:empirical}. The table reports $t$ statistics in parentheses. In square brackets we report for each parameter estimate the ratio of IV standard error to  Newton-step standard error, and find that this difference can be as great as 12.53\%. Thus the iteration scheme we propose can lead to more accurate inference in practice as the estimates are more precise.

As far as the $\beta_i$ are concerned, our simulations generally show that the efficiency gains are smaller for these as compared to the $\lambda_i$. Table \ref{table:empirical_betas} reports standard error ratios and absolute t-statistics for exclusion tests and confirms this. Indeed, all standard error ratios are very close to unity and the t-statistics are practically identical. We note that our proposed iteration does not make the estimation precision of the $\beta_i$ worse and improves the estimation precision of the $\lambda_i$, leading to an overall improvement in estimation quality. 

We give a very brief description of the explanatory variables in $X_n$  and refer the reader to KKM for details. The covariates include the number of proximate VCFs and DBFs to capture the effects of being in areas of high VCF or DBF concentration. Firm-specific characteristics include the distance from each DBF to its funding VCFs, , the average age of each funding
VCF, exposure of VCFs through
syndication and an indicator for foreign VCF investment. Variables controlling for DBF-specific factors include firm age, dummies for receiving a grant and being in an R\&D tax credit state, a cost of business index for the DBF's home state, distance to the closest university and the number of non-biotech establishments in the DBF's zip code. Two further variables recognize that additional factors can affect the cost
of doing business in ways that influence the VC funding levels of
a given DBF.

\clearpage

\setcounter{table}{0}
\renewcommand\thetable{\arabic{table}}

\begin{table}

\footnotesize

\hspace*{-0.6cm}
\begin{tabular}{c|cccc|cccc|cccc}
&\multicolumn{4}{c}{$n=200$}&\multicolumn{4}{c}{$n=400$}&\multicolumn{4}{c}{$n=800$}\\
\midrule
$p=2$&$\hat{\theta}_n$&$\hat{\hat{\theta}}_n^{1}$&$\hat{\hat{\theta}}_n^{3}$&$\hat{\hat{\theta}}_n^{6}$&$\hat{\theta}_n$&$\hat{\hat{\theta}}_n^{1}$&$\hat{\hat{\theta}}_n^{3}$&$\hat{\hat{\theta}}_n^{6}$&$\hat{\theta}_n$&$\hat{\hat{\theta}}_n^{1}$&$\hat{\hat{\theta}}_n^{3}$&$\hat{\hat{\theta}}_n^{6}$\\
\midrule
$\lambda_1$&    0.4488 & 0.4381 & 0.4153 & 0.4148 & 0.4198 & 0.4144 & 0.4037 & 0.4037 & 0.4251 & 0.4061 & 0.3997 & 0.3997 \\ 
                                                                                                           
$\lambda_2$&    0.4582 & 0.4611 & 0.4812 & 0.4818 & 0.4834 & 0.4843 & 0.4942 & 0.4942 & 0.4769 & 0.4934 & 0.4992 & 0.4992 \\ 
                                                                                                           
$\beta_1$ &   0.9265 & 0.9922 & 1.0164 & 1.0160 & 0.9612 & 1.0002 & 1.0083 & 1.0082 & 0.9832 & 1.0052 & 1.0100 & 1.0101 \\ 
                                                                                                           
$\beta_2$ &   0.4586 & 0.5061 & 0.5218 & 0.5211 & 0.4858 & 0.5120 & 0.5168 & 0.5166 & 0.4862 & 0.5011 & 0.5044 & 0.5044 \\ 
\midrule
                                         $p=4$\\
                                         \midrule                                                       
$\lambda_1$ &   0.3021 & 0.3006 & 0.3008 & 0.2999 & 0.3181 & 0.3085 & 0.3079 & 0.3110 & 0.3237 & 0.3010 & 0.2975 & 0.2975 \\
\bottomrule
\end{tabular}
\caption{Monte Carlo mean of parameter estimates. IV and iterated Newton-step estimates with $N(0,1)$ errors. Bounded $h_n$. The parameter values used in the DGPs are as follows. $\beta_1=1$, $\beta_2=0.5$ and $p=2:\lambda_1=0.4,\lambda_2=0.5;p=4: \lambda_1=0.3,\lambda_i=0.2,i=2,3,4;p=6: \lambda_i=0.15,i=1,\ldots,6.$}\label{table:norm_mean}

\end{table}

 \begin{table}[t]

\footnotesize

\hspace*{-0.6cm}\begin{tabular}{c|cccc|cccc|cccc}
&\multicolumn{4}{c}{$n=200$}&\multicolumn{4}{c}{$n=400$}&\multicolumn{4}{c}{$n=800$}\\
\midrule
$p=2$&$\hat{\theta}_n$&$\hat{\hat{\theta}}_n^{1}$&$\hat{\hat{\theta}}_n^{3}$&$\hat{\hat{\theta}}_n^{6}$&$\hat{\theta}_n$&$\hat{\hat{\theta}}_n^{1}$&$\hat{\hat{\theta}}_n^{3}$&$\hat{\hat{\theta}}_n^{6}$&$\hat{\theta}_n$&$\hat{\hat{\theta}}_n^{1}$&$\hat{\hat{\theta}}_n^{3}$&$\hat{\hat{\theta}}_n^{6}$\\
\midrule
  $\lambda_1$&  0.4715 & 0.4620 & 0.4266 & 0.4216 & 0.4728 & 0.4558 & 0.4325 & 0.4305 & 0.4055 & 0.4044 & 0.3978 & 0.3978 \\   
                                                                                                             
      $\lambda_2$&    0.4378 & 0.4402 & 0.4712 & 0.4761 & 0.4332 & 0.4454 & 0.4665 & 0.4685 & 0.4978 & 0.4951 & 0.5009 & 0.5010 \\   
                                                                                                             
$\beta_1$&    0.8999 & 0.9659 & 1.0062 & 1.0069 & 0.9413 & 0.9864 & 1.0050 & 1.0057 & 0.9699 & 1.0035 & 1.0107 & 1.0107 \\   
                                                                                                             
  $\beta_2$&  0.4507 & 0.4935 & 0.5174 & 0.5175 & 0.4591 & 0.4879 & 0.4998 & 0.5003 & 0.4770 & 0.4987 & 0.5023 & 0.5024 \\   
                                                                                         \midrule
                                         $p=4$\\
                                         \midrule                                                                                  
      $\lambda_1$&0.3136 & 0.3104 & 0.3118 & 0.3063 & 0.3501 & 0.3297 & 0.3165 & 0.3146 & 0.2934 & 0.2999 & 0.3021 & 0.3021 \\
\bottomrule
\end{tabular}
\caption{Monte Carlo mean of parameter estimates. IV and iterated Newton-step estimates with $t_6$ errors. Bounded $h_n$. The parameter values used in the DGPs are as follows. $\beta_1=1$, $\beta_2=0.5$ and $p=2:\lambda_1=0.4,\lambda_2=0.5;p=4: \lambda_1=0.3,\lambda_i=0.2,i=2,3,4;p=6: \lambda_i=0.15,i=1,\ldots,6.$}\label{table:t6_mean}

\end{table}
 \begin{table}[t]

\footnotesize

\hspace*{-0.6cm}\begin{tabular}{c|cccc|cccc|cccc}
&\multicolumn{4}{c}{$n=200$}&\multicolumn{4}{c}{$n=400$}&\multicolumn{4}{c}{$n=800$}\\
\midrule
$p=2$&$\hat{\theta}_n$&$\hat{\hat{\theta}}_n^{1}$&$\hat{\hat{\theta}}_n^{3}$&$\hat{\hat{\theta}}_n^{6}$&$\hat{\theta}_n$&$\hat{\hat{\theta}}_n^{1}$&$\hat{\hat{\theta}}_n^{3}$&$\hat{\hat{\theta}}_n^{6}$&$\hat{\theta}_n$&$\hat{\hat{\theta}}_n^{1}$&$\hat{\hat{\theta}}_n^{3}$&$\hat{\hat{\theta}}_n^{6}$\\
\midrule
$\lambda_1$&    0.2708 & 0.0920 & 0.0564 & 0.0581 & 0.1338 & 0.0232 & 0.0160 & 0.0162 & 0.0595 & 0.0041 & 0.0026 & 0.0026 \\
                                                                                                          
$\lambda_2$&    0.2688 & 0.0890 & 0.0541 & 0.0544 & 0.1315 & 0.0220 & 0.0147 & 0.0147 & 0.0585 & 0.0040 & 0.0027 & 0.0027 \\
                                                                                                          
$\beta_1$&    0.0848 & 0.0593 & 0.0541 & 0.0571 & 0.0394 & 0.0257 & 0.0268 & 0.0277 & 0.0198 & 0.0124 & 0.0119 & 0.0119 \\
                                                                                                          
$\beta_2$&    0.0659 & 0.0490 & 0.0483 & 0.0517 & 0.0332 & 0.0234 & 0.0240 & 0.0246 & 0.0168 & 0.0115 & 0.0112 & 0.0112 \\
    \midrule
                                         $p=4$\\
                                         \midrule                                                                                                      
$\lambda_1$&    0.4368 & 0.1998 & 0.1416 & 0.1418 & 0.2197 & 0.0788 & 0.1164 & 0.1895 & 0.0862 & 0.0065 & 0.0035 & 0.0035 \\
\bottomrule
\end{tabular}
\caption{Monte Carlo mean squared error of parameter estimates. IV and iterated Newton-step estimates with $N(0,1)$ errors. Bounded $h_n$. The parameter values used in the DGPs are as follows. $\beta_1=1$, $\beta_2=0.5$ and $p=2:\lambda_1=0.4,\lambda_2=0.5;p=4: \lambda_1=0.3,\lambda_i=0.2,i=2,3,4;p=6: \lambda_i=0.15,i=1,\ldots,6.$}\label{table:norm_rmse}

\end{table}

     \begin{table}[t]

\footnotesize

\begin{center}

\hspace*{-0.6cm}\begin{tabular}{c|cccc|cccc|cccc}
&\multicolumn{4}{c}{$n=200$}&\multicolumn{4}{c}{$n=400$}&\multicolumn{4}{c}{$n=800$}\\
\midrule
$p=2$&$\hat{\theta}_n$&$\hat{\hat{\theta}}_n^{1}$&$\hat{\hat{\theta}}_n^{3}$&$\hat{\hat{\theta}}_n^{6}$&$\hat{\theta}_n$&$\hat{\hat{\theta}}_n^{1}$&$\hat{\hat{\theta}}_n^{3}$&$\hat{\hat{\theta}}_n^{6}$&$\hat{\theta}_n$&$\hat{\hat{\theta}}_n^{1}$&$\hat{\hat{\theta}}_n^{3}$&$\hat{\hat{\theta}}_n^{6}$\\
\midrule

$\lambda_1$& 0.4325 & 0.2335 & 0.1675 & 0.1642 & 0.2502 & 0.0999 & 0.0765 & 0.0755 & 0.1000 & 0.0095 & 0.0030 & 0.0030 \\  
                                                                                                            
$\lambda_2$&    0.4266 & 0.2262 & 0.1596 & 0.1556 & 0.2427 & 0.0936 & 0.0707 & 0.0696 & 0.0992 & 0.0092 & 0.0031 & 0.0031 \\  
                                                                                                            
$\beta_1$&    0.1355 & 0.1005 & 0.0966 & 0.0989 & 0.0650 & 0.0472 & 0.0477 & 0.0480 & 0.0295 & 0.0179 & 0.0168 & 0.0168 \\  
                                                                                                            
$\beta_2$&    0.0989 & 0.0778 & 0.0826 & 0.0848 & 0.0539 & 0.0383 & 0.0387 & 0.0390 & 0.0251 & 0.0164 & 0.0156 & 0.0156 \\  
              \midrule
                                         $p=4$\\
                                         \midrule                                                                                                                    
$\lambda_1$&    0.6466 & 0.3928 & 0.3568 & 0.3974 & 0.3465 & 0.1219 & 0.0788 & 0.0788 & 0.1533 & 0.0393 & 0.0327 & 0.0328 \\
\bottomrule
\end{tabular}
\caption{Monte Carlo mean squared error of parameter estimates. IV and iterated Newton-step estimates with $t_6$ errors. Bounded $h_n$. The parameter values used in the DGPs are as follows. $\beta_1=1$, $\beta_2=0.5$ and $p=2:\lambda_1=0.4,\lambda_2=0.5;p=4: \lambda_1=0.3,\lambda_i=0.2,i=2,3,4;p=6: \lambda_i=0.15,i=1,\ldots,6.$}\label{table:t6_rmse}
\end{center}
\end{table}  
 \begin{table}[t]

\footnotesize

\begin{center}

\hspace*{-0.7cm}\begin{tabular}{c|ccc|ccc|ccc}
&\multicolumn{3}{c}{$n=200$}&\multicolumn{3}{c}{$n=400$}&\multicolumn{3}{c}{$n=800$}\\
\midrule                                   $p=2$&$\frac{\text{RMSE} \left(\hat{\theta}_n\right)}{\text{RMSE}\left(\hat{\hat{\theta}}_n^{1}\right)}$&$\frac{\text{RMSE} \left(\hat{\theta}_n\right)}{\text{RMSE}\left(\hat{\hat{\theta}}_n^{3}\right)}$&$\frac{\text{RMSE} \left(\hat{\theta}_n\right)}{\text{RMSE}\left(\hat{\hat{\theta}}_n^{6}\right)}$&$\frac{\text{RMSE} \left(\hat{\theta}_n\right)}{\text{RMSE}\left(\hat{\hat{\theta}}_n^{1}\right)}$&$\frac{\text{RMSE} \left(\hat{\theta}_n\right)}{\text{RMSE}\left(\hat{\hat{\theta}}_n^{3}\right)}$&$\frac{\text{RMSE} \left(\hat{\theta}_n\right)}{\text{RMSE}\left(\hat{\hat{\theta}}_n^{6}\right)}$&$\frac{\text{RMSE} \left(\hat{\theta}_n\right)}{\text{RMSE}\left(\hat{\hat{\theta}}_n^{1}\right)}$&$\frac{\text{RMSE} \left(\hat{\theta}_n\right)}{\text{RMSE}\left(\hat{\hat{\theta}}_n^{3}\right)}$&$\frac{\text{RMSE} \left(\hat{\theta}_n\right)}{\text{RMSE}\left(\hat{\hat{\theta}}_n^{6}\right)}$\\
\midrule              
$\lambda_1$&    1.7153 &2.1714 & 2.1687 & 2.4010 & 2.9156 & 2.8782 & 3.7949 & 4.7436 & 4.7428 \\
                                                                                             
$\lambda_2$&    1.7376 &2.2054 & 2.2259 & 2.4464 & 2.9881 & 2.9891 & 3.8068 & 4.6181 & 4.6176 \\
                                                                                             
$\beta_1$&    1.1957 &  1.2528 & 1.2263 & 1.2386 & 1.2355 & 1.1936 & 1.2631 & 1.2884 & 1.2882 \\
                                                                                             
$\beta_2$&    1.1593 &  1.1720 & 1.1398 & 1.1907 & 1.1912 & 1.1630 & 1.2086 & 1.2257 & 1.2256 \\
                                \midrule
                                         $p=4$\\
\bottomrule
\end{tabular}
\caption{Monte Carlo relative root MSE of IV estimates to iterated Newton-step estimates with $N(0,1)$ errors. Bounded $h_n$. The parameter values used in the DGPs are as follows. $\beta_1=1$, $\beta_2=0.5$ and $p=2:\lambda_1=0.4,\lambda_2=0.5;p=4: \lambda_1=0.3,\lambda_i=0.2,i=2,3,4;p=6: \lambda_i=0.15,i=1,\ldots,6.$}\label{table:normal_rrmse}
\end{center}
\end{table}

 \begin{table}[t]

\footnotesize

\begin{center}

\hspace*{-0.7cm}\begin{tabular}{c|ccc|ccc|ccc}
&\multicolumn{3}{c}{$n=200$}&\multicolumn{3}{c}{$n=400$}&\multicolumn{3}{c}{$n=800$}\\
\midrule                                                                 $p=2$&$\frac{\text{RMSE} \left(\hat{\theta}_n\right)}{\text{RMSE}\left(\hat{\hat{\theta}}_n^{1}\right)}$&$\frac{\text{RMSE} \left(\hat{\theta}_n\right)}{\text{RMSE}\left(\hat{\hat{\theta}}_n^{3}\right)}$&$\frac{\text{RMSE} \left(\hat{\theta}_n\right)}{\text{RMSE}\left(\hat{\hat{\theta}}_n^{6}\right)}$&$\frac{\text{RMSE} \left(\hat{\theta}_n\right)}{\text{RMSE}\left(\hat{\hat{\theta}}_n^{1}\right)}$&$\frac{\text{RMSE} \left(\hat{\theta}_n\right)}{\text{RMSE}\left(\hat{\hat{\theta}}_n^{3}\right)}$&$\frac{\text{RMSE} \left(\hat{\theta}_n\right)}{\text{RMSE}\left(\hat{\hat{\theta}}_n^{6}\right)}$&$\frac{\text{RMSE} \left(\hat{\theta}_n\right)}{\text{RMSE}\left(\hat{\hat{\theta}}_n^{1}\right)}$&$\frac{\text{RMSE} \left(\hat{\theta}_n\right)}{\text{RMSE}\left(\hat{\hat{\theta}}_n^{3}\right)}$&$\frac{\text{RMSE} \left(\hat{\theta}_n\right)}{\text{RMSE}\left(\hat{\hat{\theta}}_n^{6}\right)}$\\
\midrule                                                                
$\lambda_1$&    1.3611 & 1.5838 & 1.6234 & 1.5828 & 1.7797 & 1.8209 & 3.2483 & 5.7778 & 5.7832 \\
                                                                                              
$\lambda_2$&    1.3732 & 1.6083 & 1.6552 & 1.6106 & 1.8208 & 1.8673 & 3.2901 & 5.6271 & 5.6293 \\
                                                                                              
$\beta_1$&    1.1611 &   1.1917 & 1.1769 & 1.1740 & 1.1715 & 1.1650 & 1.2833 & 1.3258 & 1.3261 \\
                                                                                              
$\beta_2$&    1.1274 &   1.1032 & 1.0848 & 1.1867 & 1.1876 & 1.1758 & 1.2365 & 1.2693 & 1.2701 \\
                                                           \midrule
                                         $p=4$\\
\bottomrule
\end{tabular}
\caption{Monte Carlo relative root MSE of IV estimates to iterated Newton-step estimates with $t_6$ errors. Bounded $h_n$. The parameter values used in the DGPs are as follows. $\beta_1=1$, $\beta_2=0.5$ and $p=2:\lambda_1=0.4,\lambda_2=0.5;p=4: \lambda_1=0.3,\lambda_i=0.2,i=2,3,4;p=6: \lambda_i=0.15,i=1,\ldots,6.$}\label{table:t6_rrmse}
\end{center}
\end{table}
 \begin{table}[t]

\footnotesize

\begin{center}

\hspace*{-0.7cm}\begin{tabular}{c|ccc|ccc|ccc}
&\multicolumn{3}{c}{$n=200$}&\multicolumn{3}{c}{$n=400$}&\multicolumn{3}{c}{$n=800$}\\
\midrule                                                                 $p=2$&$\frac{\text{RMSE} \left(\mathring{\theta}_n\right)}{\text{RMSE}\left(\hat{\hat{\theta}}_n^{1}\right)}$&$\frac{\text{RMSE} \left(\mathring{\theta}_n\right)}{\text{RMSE}\left(\hat{\hat{\theta}}_n^{3}\right)}$&$\frac{\text{RMSE} \left(\mathring{\theta}_n\right)}{\text{RMSE}\left(\hat{\hat{\theta}}_n^{6}\right)}$&$\frac{\text{RMSE} \left(\mathring{\theta}_n\right)}{\text{RMSE}\left(\hat{\hat{\theta}}_n^{1}\right)}$&$\frac{\text{RMSE} \left(\mathring{\theta}_n\right)}{\text{RMSE}\left(\hat{\hat{\theta}}_n^{3}\right)}$&$\frac{\text{RMSE} \left(\mathring{\theta}_n\right)}{\text{RMSE}\left(\hat{\hat{\theta}}_n^{6}\right)}$&$\frac{\text{RMSE} \left(\mathring{\theta}_n\right)}{\text{RMSE}\left(\hat{\hat{\theta}}_n^{1}\right)}$&$\frac{\text{RMSE} \left(\mathring{\theta}_n\right)}{\text{RMSE}\left(\hat{\hat{\theta}}_n^{3}\right)}$&$\frac{\text{RMSE} \left(\mathring{\theta}_n\right)}{\text{RMSE}\left(\hat{\hat{\theta}}_n^{6}\right)}$\\
\midrule                                        
$\lambda_1$&   0.3415 & 0.4322 & 0.4317 & 0.4811 & 0.5842 & 0.5767 & 0.8024 & 1.0030 & 1.0028 \\
                                                                                                                             
$\lambda_2$&0.3562 &    0.4521 & 0.4563 & 0.5027 & 0.6141 & 0.6143 & 0.8269 & 1.0031 & 1.0030 \\
                                                                                                                             
$\beta_1$&  0.9036 &    0.9468 & 0.9267 & 0.9432 & 0.9409 & 0.9090 & 0.9788 & 0.9984 & 0.9983 \\
                                                                                                                             
$\beta_2$&0.9637 &              0.9743 & 0.9475 & 0.9725 & 0.9729 & 0.9499 & 0.9851 & 0.9991 & 0.9989 \\
\bottomrule
\end{tabular}
\caption{Monte Carlo relative root MSE of ML estimates to iterated Newton-step estimates with $N(0,1)$ errors. Bounded $h_n$. The parameter values used in the DGPs are as follows. $\beta_1=1$, $\beta_2=0.5$ and $p=2:\lambda_1=0.4,\lambda_2=0.5;p=4: \lambda_1=0.3,\lambda_i=0.2,i=2,3,4;p=6: \lambda_i=0.15,i=1,\ldots,6.$}\label{table:normal_mlerrmse}
\end{center}
\end{table}

 \begin{table}[t]

\footnotesize

\begin{center}

\begin{tabular}{c|ccc|ccc|ccc}
&\multicolumn{3}{c}{$n=200$}&\multicolumn{3}{c}{$n=400$}&\multicolumn{3}{c}{$n=800$}\\
\midrule                                                                 $p=2$&$\hat\theta_n$&$\hat{\hat\theta}_n^1$&$\hat{\hat\theta}_n^3$&$\hat\theta_n$&$\hat{\hat\theta}_n^1$&$\hat{\hat\theta}_n^3$&$\hat\theta_n$&$\hat{\hat\theta}_n^1$&$\hat{\hat\theta}_n^3$\\
\midrule                                        
$\lambda_1$&      0.3992 & 0.3934 & 0.3933 & 0.3985 & 0.3952 & 0.3952 & 0.4012 & 0.3988 & 0.3988 \\  
                                                                                                
$\lambda_2$&      0.5019 & 0.4945 & 0.4944 & 0.5015 & 0.4984 & 0.4984 & 0.4974 & 0.4967 & 0.4967 \\  
                                                                                                
$\beta_1$&        0.9892 & 1.0241 & 1.0245 & 0.9958 & 1.0145 & 1.0145 & 1.0040 & 1.0137 & 1.0137 \\  
                                                                                                
$\beta_2$&        0.4875 & 0.5207 & 0.5210 & 0.4985 & 0.5171 & 0.5171 & 0.5027 & 0.5120 & 0.5120 \\  
                                                                                                           \midrule
                                         $p=4$\\
                                         \midrule                                                                        
$\lambda_1$&      0.3002 & 0.2952 & 0.2952 & 0.2993 & 0.2966 & 0.2966 & 0.2981 & 0.2982 & 0.2982 \\
\bottomrule
\end{tabular}
\caption{Monte Carlo mean of parameter estimates. IV and iterated Newton-step estimates with $N(0,1)$ errors. Divergent $h_n$. The parameter values used in the DGPs are as follows. $\beta_1=1$, $\beta_2=0.5$ and $p=2:\lambda_1=0.4,\lambda_2=0.5;p=4: \lambda_1=0.3,\lambda_i=0.2,i=2,3,4;p=6: \lambda_i=0.15,i=1,\ldots,6.$}\label{table:hdiviv_norm_mean}
\end{center}
\end{table}

 \begin{table}[t]

\footnotesize

\begin{center}

\begin{tabular}{c|ccc|ccc|ccc}
&\multicolumn{3}{c}{$n=200$}&\multicolumn{3}{c}{$n=400$}&\multicolumn{3}{c}{$n=800$}\\
\midrule                                                                 $p=2$&$\tilde\theta_n$&$\tilde{\tilde\theta}_n^1$&$\tilde{\tilde\theta}_n^3$&$\tilde\theta_n$&$\tilde{\tilde\theta}_n^1$&$\tilde{\tilde\theta}_n^3$&$\tilde\theta_n$&$\tilde{\tilde\theta}_n^1$&$\tilde{\tilde\theta}_n^3$\\
\midrule                                        
$\lambda_1$&      0.4017 & 0.3956 & 0.3933 & 0.3967 & 0.3936 & 0.3952 & 0.3977 & 0.3956 & 0.3988 \\  
                                                                                                
$\lambda_2$&      0.5093 & 0.5014 & 0.4944 & 0.5067 & 0.5034 & 0.4984 & 0.5021 & 0.5012 & 0.4967 \\  
                                                                                                
$\beta_1$&        0.9643 & 1.0010 & 1.0245 & 0.9854 & 1.0045 & 1.0145 & 0.9998 & 1.0096 & 1.0137 \\  
                                                                                                
$\beta_2$&        0.4635 & 0.4983 & 0.5210 & 0.4889 & 0.5078 & 0.5171 & 0.4990 & 0.5083 & 0.5120 \\  
                                        \midrule
                                         $p=4$\\
                                         \midrule                                                                                                                                           
$\lambda_1$&      0.3062 & 0.3008 & 0.2952 & 0.3035 & 0.3006 & 0.2966 & 0.3034 & 0.3032 & 0.2982 \\
\bottomrule
\end{tabular}
\caption{Monte Carlo mean of parameter estimates. OLS and iterated Newton-step estimates with $N(0,1)$ errors. Divergent $h_n$. The parameter values used in the DGPs are as follows. $\beta_1=1$, $\beta_2=0.5$ and $p=2:\lambda_1=0.4,\lambda_2=0.5;p=4: \lambda_1=0.3,\lambda_i=0.2,i=2,3,4;p=6: \lambda_i=0.15,i=1,\ldots,6.$}\label{table:hdivols_norm_mean}
\end{center}
\end{table}

 \begin{table}[t]

\footnotesize

\begin{center}

\begin{tabular}{c|ccc|ccc|ccc}
&\multicolumn{3}{c}{$n=200$}&\multicolumn{3}{c}{$n=400$}&\multicolumn{3}{c}{$n=800$}\\
\midrule                                                                 $p=2$&$\hat\theta_n$&$\hat{\hat\theta}_n^1$&$\hat{\hat\theta}_n^3$&$\hat\theta_n$&$\hat{\hat\theta}_n^1$&$\hat{\hat\theta}_n^3$&$\hat\theta_n$&$\hat{\hat\theta}_n^1$&$\hat{\hat\theta}_n^3$\\
\midrule                                        
$\lambda_1$&      0.0051 & 0.0040 & 0.0040 & 0.0046 & 0.0037 & 0.0037 & 0.0063 & 0.0051 & 0.0051 \\  
                                                                                                
$\lambda_2$&      0.0053 & 0.0041 & 0.0041 & 0.0046 & 0.0037 & 0.0037 & 0.0062 & 0.0049 & 0.0049 \\  
                                                                                                
$\beta_1$&        0.0571 & 0.0541 & 0.0540 & 0.0290 & 0.0284 & 0.0284 & 0.0151 & 0.0151 & 0.0151 \\  
                                                                                                
$\beta_2$&        0.0532 & 0.0498 & 0.0497 & 0.0286 & 0.0280 & 0.0280 & 0.0145 & 0.0143 & 0.0143 \\  
                                    \midrule
                                         $p=4$\\
                                         \midrule                                                                                                                                               
$\lambda_1$&      0.0087 & 0.0063 & 0.0063 & 0.0075 & 0.0060 & 0.0060 & 0.0082 & 0.0065 & 0.0065 \\
\bottomrule
\end{tabular}
\caption{Monte Carlo mean squared error of parameter estimates. IV and iterated Newton-step estimates with $N(0,1)$ errors. Divergent $h_n$. The parameter values used in the DGPs are as follows. $\beta_1=1$, $\beta_2=0.5$ and $p=2:\lambda_1=0.4,\lambda_2=0.5;p=4: \lambda_1=0.3,\lambda_i=0.2,i=2,3,4;p=6: \lambda_i=0.15,i=1,\ldots,6.$}\label{table:hdiviv_norm_rmse}
\end{center}
\end{table}

 \begin{table}[t]

\footnotesize

\begin{center}

\begin{tabular}{c|ccc|ccc|ccc}
&\multicolumn{3}{c}{$n=200$}&\multicolumn{3}{c}{$n=400$}&\multicolumn{3}{c}{$n=800$}\\
\midrule                                                                 $p=2$&$\tilde\theta_n$&$\tilde{\tilde\theta}_n^1$&$\tilde{\tilde\theta}_n^3$&$\tilde\theta_n$&$\tilde{\tilde\theta}_n^1$&$\tilde{\tilde\theta}_n^3$&$\tilde\theta_n$&$\tilde{\tilde\theta}_n^1$&$\tilde{\tilde\theta}_n^3$\\
\midrule                                        
$\lambda_1$&      0.0045 & 0.0052 & 0.0040 & 0.0042 & 0.0051 & 0.0037 & 0.0057 & 0.0066 & 0.0051 \\  
                                                                                                
$\lambda_2$&      0.0046 & 0.0050 & 0.0041 & 0.0042 & 0.0051 & 0.0037 & 0.0055 & 0.0063 & 0.0049 \\  
                                                                                                
$\beta_1$&        0.0574 & 0.0547 & 0.0540 & 0.0291 & 0.0286 & 0.0284 & 0.0151 & 0.0151 & 0.0151 \\  
                                                                                                
$\beta_2$&        0.0527 & 0.0496 & 0.0497 & 0.0286 & 0.0281 & 0.0280 & 0.0144 & 0.0143 & 0.0143 \\  
                                    \midrule
                                         $p=4$\\
                                         \midrule                                                                                                                                               
$\lambda_1$&      0.0074 & 0.0085 & 0.0063 & 0.0069 & 0.0085 & 0.0060 & 0.0074 & 0.0095 & 0.0065 \\
\bottomrule
\end{tabular}
\caption{Monte Carlo mean squared error of parameter estimates. OLS and iterated Newton-step estimates with $N(0,1)$ errors. Divergent $h_n$. The parameter values used in the DGPs are as follows. $\beta_1=1$, $\beta_2=0.5$ and $p=2:\lambda_1=0.4,\lambda_2=0.5;p=4: \lambda_1=0.3,\lambda_i=0.2,i=2,3,4;p=6: \lambda_i=0.15,i=1,\ldots,6.$}\label{table:hdivols_norm_rmse}
\end{center}
\end{table}
\clearpage
\begin{landscape}
 \begin{table}[t]

\footnotesize

\begin{center}

\begin{tabular}{c|cccc|cccc|cccc}
&\multicolumn{4}{c}{$n=200$}&\multicolumn{4}{c}{$n=400$}&\multicolumn{4}{c}{$n=800$}\\
\midrule                                   $p=2$&$\frac{\text{RMSE} \left(\hat{\theta}_n\right)}{\text{RMSE}\left(\hat{\hat{\theta}}_n^{1}\right)}$&$\frac{\text{RMSE} \left(\hat{\theta}_n\right)}{\text{RMSE}\left(\hat{\hat{\theta}}_n^{3}\right)}$&$\frac{\text{RMSE} \left(\tilde{\theta}_n\right)}{\text{RMSE}\left(\tilde{\tilde{\theta}}_n^{1}\right)}$&$\frac{\text{RMSE} \left(\tilde{\theta}_n\right)}{\text{RMSE}\left(\tilde{\tilde{\theta}}_n^{3}\right)}$&$\frac{\text{RMSE} \left(\hat{\theta}_n\right)}{\text{RMSE}\left(\hat{\hat{\theta}}_n^{1}\right)}$&$\frac{\text{RMSE} \left(\hat{\theta}_n\right)}{\text{RMSE}\left(\hat{\hat{\theta}}_n^{3}\right)}$&$\frac{\text{RMSE} \left(\tilde{\theta}_n\right)}{\text{RMSE}\left(\tilde{\tilde{\theta}}_n^{1}\right)}$&$\frac{\text{RMSE} \left(\tilde{\theta}_n\right)}{\text{RMSE}\left(\tilde{\tilde{\theta}}_n^{3}\right)}$&$\frac{\text{RMSE} \left(\hat{\theta}_n\right)}{\text{RMSE}\left(\hat{\hat{\theta}}_n^{1}\right)}$&$\frac{\text{RMSE} \left(\hat{\theta}_n\right)}{\text{RMSE}\left(\hat{\hat{\theta}}_n^{3}\right)}$&$\frac{\text{RMSE} \left(\tilde{\theta}_n\right)}{\text{RMSE}\left(\tilde{\tilde{\theta}}_n^{1}\right)}$&$\frac{\text{RMSE} \left(\tilde{\theta}_n\right)}{\text{RMSE}\left(\tilde{\tilde{\theta}}_n^{3}\right)}$\\
\midrule              
$\lambda_1$&      1.1211 & 1.1214 & 0.9365 & 1.0612 & 1.1141 & 1.1141 & 0.9056 & 1.0592 & 1.1172 & 1.1172 & 0.9312 & 1.0634 \\  
                                                                                                                           
$\lambda_2$&      1.1326 & 1.1331 & 0.9592 & 1.0643 & 1.1115 & 1.1115 & 0.9072 & 1.0646 & 1.1275 & 1.1275 & 0.9328 & 1.0629 \\
\bottomrule
\end{tabular}
\caption{Monte Carlo relative root MSE of parameter estimates. IV, OLS and iterated Newton-step estimates with $N(0,1)$ errors. Divergent $h_n$. The parameter values used in the DGPs are as follows. $\beta_1=1$, $\beta_2=0.5$ and $p=2:\lambda_1=0.4,\lambda_2=0.5;p=4: \lambda_1=0.3,\lambda_i=0.2,i=2,3,4;p=6: \lambda_i=0.15,i=1,\ldots,6.$}\label{table:hdiv_norm_rrmse}
\end{center}
\end{table}
\end{landscape}
\clearpage
\begin{table}[t]

\footnotesize

\begin{center}

\hspace*{-0.6cm}\begin{tabular}{c|cccc|cccc|cccc}
&\multicolumn{4}{c}{$n=200$}&\multicolumn{4}{c}{$n=400$}&\multicolumn{4}{c}{$n=800$}\\
\midrule
$p=2$&$\hat{\theta}_n$&$\hat{\hat{\theta}}_n^{1}$&$\hat{\hat{\theta}}_n^{3}$&$\hat{\hat{\theta}}_n^{6}$&$\hat{\theta}_n$&$\hat{\hat{\theta}}_n^{1}$&$\hat{\hat{\theta}}_n^{3}$&$\hat{\hat{\theta}}_n^{6}$&$\hat{\theta}_n$&$\hat{\hat{\theta}}_n^{1}$&$\hat{\hat{\theta}}_n^{3}$&$\hat{\hat{\theta}}_n^{6}$\\
\midrule
 $\lambda_1$&     0.4503 & 0.4406 & 0.4215 & 0.4191 & 0.4223 & 0.4171 & 0.4057 & 0.4053 & 0.4229 & 0.4070 & 0.4007 & 0.4005 \\   
                                                                                                                            
 $\lambda_2$&     0.4572 & 0.4588 & 0.4757 & 0.4781 & 0.4815 & 0.4817 & 0.4923 & 0.4926 & 0.4793 & 0.4925 & 0.4983 & 0.4984 \\   
                                                                                                                            
 $\beta_1$&       0.9308 & 0.9957 & 1.0149 & 1.0151 & 0.9592 & 1.0003 & 1.0084 & 1.0086 & 0.9835 & 1.0059 & 1.0101 & 1.0103 \\   
                                                                                                                            
 $\beta_2$&       0.4583 & 0.5061 & 0.5176 & 0.5169 & 0.4850 & 0.5134 & 0.5177 & 0.5179 & 0.4872 & 0.5023 & 0.5050 & 0.5052 \\ 
          \midrule
                                         $p=4$\\
                                         \midrule 
                                                                                                                            
 $\lambda_1$&     0.3051 & 0.3085 & 0.3133 & 0.3118 & 0.3220 & 0.3063 & 0.3034 & 0.3016 & 0.3204 & 0.3044 & 0.3029 & 0.3027 \\
\bottomrule
\end{tabular}
\caption{Monte Carlo mean of parameter estimates. IV and iterated Newton-step estimates with heteroskedastic zero mean Gaussian errors. Bounded $h_n$. The parameter values used in the DGPs are as follows. $\beta_1=1$, $\beta_2=0.5$ and $p=2:\lambda_1=0.4,\lambda_2=0.5;p=4: \lambda_1=0.3,\lambda_i=0.2,i=2,3,4;p=6: \lambda_i=0.15,i=1,\ldots,6.$}\label{table:hetnorm_mean}
\end{center}

\end{table}
\begin{table}[t]

\footnotesize

\begin{center}

\hspace*{-0.6cm}\begin{tabular}{c|cccc|cccc|cccc}
&\multicolumn{4}{c}{$n=200$}&\multicolumn{4}{c}{$n=400$}&\multicolumn{4}{c}{$n=800$}\\
\midrule
$p=2$&$\hat{\theta}_n$&$\hat{\hat{\theta}}_n^{1}$&$\hat{\hat{\theta}}_n^{3}$&$\hat{\hat{\theta}}_n^{6}$&$\hat{\theta}_n$&$\hat{\hat{\theta}}_n^{1}$&$\hat{\hat{\theta}}_n^{3}$&$\hat{\hat{\theta}}_n^{6}$&$\hat{\theta}_n$&$\hat{\hat{\theta}}_n^{1}$&$\hat{\hat{\theta}}_n^{3}$&$\hat{\hat{\theta}}_n^{6}$\\
\midrule
 $\lambda_1$&     0.2641 & 0.0959 & 0.0713 & 0.0719 & 0.1396 & 0.0290 & 0.0205 & 0.0206 & 0.0592 & 0.0050 & 0.0028 & 0.0028 \\  
                                                                                                                           
 $\lambda_2$&     0.2612 & 0.0920 & 0.0656 & 0.0657 & 0.1368 & 0.0273 & 0.0188 & 0.0188 & 0.0583 & 0.0049 & 0.0029 & 0.0029 \\  
                                                                                                                           
 $\beta_1$&       0.0854 & 0.0651 & 0.0665 & 0.0685 & 0.0401 & 0.0280 & 0.0295 & 0.0297 & 0.0193 & 0.0130 & 0.0126 & 0.0126 \\  
                                                                                                                           
 $\beta_2$&       0.0664 & 0.0527 & 0.0592 & 0.0614 & 0.0338 & 0.0256 & 0.0263 & 0.0265 & 0.0166 & 0.0122 & 0.0119 & 0.0119 \\  
                                                   \midrule
                                         $p=4$\\
                                         \midrule                                                                                   
 $\lambda_1$&     0.4143 & 0.1888 & 0.1449 & 0.1442 & 0.2023 & 0.0471 & 0.0327 & 0.0328 & 0.0898 & 0.0144 & 0.0121 & 0.0121 \\
\bottomrule
\end{tabular}
\caption{Monte Carlo mean squared error of parameter estimates. IV and iterated Newton-step estimates with heteroskedastic zero mean Gaussian errors. Bounded $h_n$. The parameter values used in the DGPs are as follows. $\beta_1=1$, $\beta_2=0.5$ and $p=2:\lambda_1=0.4,\lambda_2=0.5;p=4: \lambda_1=0.3,\lambda_i=0.2,i=2,3,4;p=6: \lambda_i=0.15,i=1,\ldots,6.$}\label{table:hetnorm_mse}
\end{center}

\end{table}
 \begin{table}[t]

\footnotesize

\begin{center}

\hspace*{-0.7cm}\begin{tabular}{c|ccc|ccc|ccc}
&\multicolumn{3}{c}{$n=200$}&\multicolumn{3}{c}{$n=400$}&\multicolumn{3}{c}{$n=800$}\\
\midrule                                   $p=2$&$\frac{\text{RMSE} \left(\hat{\theta}_n\right)}{\text{RMSE}\left(\hat{\hat{\theta}}_n^{1}\right)}$&$\frac{\text{RMSE} \left(\hat{\theta}_n\right)}{\text{RMSE}\left(\hat{\hat{\theta}}_n^{3}\right)}$&$\frac{\text{RMSE} \left(\hat{\theta}_n\right)}{\text{RMSE}\left(\hat{\hat{\theta}}_n^{6}\right)}$&$\frac{\text{RMSE} \left(\hat{\theta}_n\right)}{\text{RMSE}\left(\hat{\hat{\theta}}_n^{1}\right)}$&$\frac{\text{RMSE} \left(\hat{\theta}_n\right)}{\text{RMSE}\left(\hat{\hat{\theta}}_n^{3}\right)}$&$\frac{\text{RMSE} \left(\hat{\theta}_n\right)}{\text{RMSE}\left(\hat{\hat{\theta}}_n^{6}\right)}$&$\frac{\text{RMSE} \left(\hat{\theta}_n\right)}{\text{RMSE}\left(\hat{\hat{\theta}}_n^{1}\right)}$&$\frac{\text{RMSE} \left(\hat{\theta}_n\right)}{\text{RMSE}\left(\hat{\hat{\theta}}_n^{3}\right)}$&$\frac{\text{RMSE} \left(\hat{\theta}_n\right)}{\text{RMSE}\left(\hat{\hat{\theta}}_n^{6}\right)}$\\
\midrule              
 $\lambda_1$&     1.6599 & 1.9196 & 1.9171 & 2.1926 & 2.6162 & 2.6054 & 3.4381 & 4.5874 & 4.5872 \\  
                                                                                                
 $\lambda_2$&     1.6846 & 1.9786 & 1.9944 & 2.2392 & 2.7016 & 2.7003 & 3.4509 & 4.4707 & 4.4704 \\  
                                                                                                
 $\beta_1$&       1.1453 & 1.1525 & 1.1182 & 1.1963 & 1.1748 & 1.1621 & 1.2174 & 1.2386 & 1.2386 \\  
                                                                                                
 $\beta_2$&       1.1222 & 1.0814 & 1.0416 & 1.1488 & 1.1403 & 1.1298 & 1.1666 & 1.1803 & 1.1801 \\ 
                                                                        \midrule
                                                     $p=4$\\
\bottomrule
\end{tabular}
\caption{Monte Carlo relative root MSE of IV estimates to iterated Newton-step estimates with heteroskedastic zero mean Gaussian errors. Bounded $h_n$. The parameter values used in the DGPs are as follows. $\beta_1=1$, $\beta_2=0.5$ and $p=2:\lambda_1=0.4,\lambda_2=0.5;p=4: \lambda_1=0.3,\lambda_i=0.2,i=2,3,4;p=6: \lambda_i=0.15,i=1,\ldots,6.$}\label{table:hetnormal_rrmse}
\end{center}
\end{table}

\begin{table}[t]

\footnotesize

\begin{center}

\begin{tabular}{c|cccccc}

$p=6$&$\hat\lambda_1$&$\hat\lambda_2$&$\hat\lambda_3$&$\hat\lambda_4$&$\hat\lambda_5$&$\hat\lambda_6$\\
\midrule
&0.1228 &   0.0853  &  0.0490 &   0.0510  &  0.0249  &  0.0136\\
&(3.0525)  &  (2.1215) &   (1.1225)  &  (1.0627)  &  (0.6064)  &  (0.3450)\\
\midrule
& $\hat{\hat\lambda}_1$&$\hat{\hat\lambda}_2$&$\hat{\hat\lambda}_3$&$\hat{\hat\lambda}_4$&$\hat{\hat\lambda}_5$&$\hat{\hat\lambda}_6$\\
\midrule
& 0.0821  &  0.0848 &   0.0080  &  0.0483  &  0.0223  &  0.0171\\
& (3.3776)  &  (2.3139)  &  (1.2404)  &  (1.1628)  &  (0.6559)  &  (0.3568)\\
&  $[$1.1065]   & [1.0907] &   [1.1050] &   [1.0942] &   [1.0817]  &  [1.0341]\\
 \midrule
$p=4$&$\hat\lambda_1$&$\hat\lambda_2$&$\hat\lambda_3$&$\hat\lambda_4$\\
\midrule
&0.0982  &  0.0844  &  0.0581  &  0.0516\\
&(2.4240)  &  (2.0840)  &  (1.3298)  &  (1.1102)\\
\midrule
&$\hat{\hat\lambda}_1$&$\hat{\hat\lambda}_2$&$\hat{\hat\lambda}_3$&$\hat{\hat\lambda}_4$\\
\midrule
&0.0844  &  0.0875 &   0.0049 &   0.0555\\
&(2.7276)  &  (2.3051) &   (1.4859)&   (1.2116)\\
&$[$1.1253]  &  [1.1061]  &  [1.1174]  &  [1.0913]\\
\midrule
$p=2$&$\hat\lambda_1$&$\hat\lambda_2$\\
\midrule
&0.0984  &  0.0894\\
&(2.4879)  &  (2.1819)\\
\midrule
&$\hat{\hat\lambda}_1$&$\hat{\hat\lambda}_2$\\
\midrule
&0.0891  &  0.0882\\
&(2.7968)  &  (2.4516)\\
&$[$1.1241] &    [1.1236]\\
\bottomrule
\end{tabular}
\caption{IV and single Newton-step estimates of $\lambda_i$ in model (\ref{KKMspec}). $t$-statistics are in parentheses and the ratio of IV standard errors to Newton-step standard errors are in square brackets.}\label{table:empirical}
\end{center}
\end{table}

\begin{table}[t]

\footnotesize

\begin{center}

\begin{tabular}{c|ccc|ccc|ccc}
&\multicolumn{3}{c}{$p=2$}&\multicolumn{3}{c}{$p=4$}&\multicolumn{3}{c}{$p=6$}\\
\midrule
&$\frac{se\left(\hat\beta\right)}{se\left(\hat{\hat\beta}\right)}$&$\left\vert\frac{\hat\beta}{se\left(\hat\beta\right)}\right\vert$&$\left\vert\frac{\hat{\hat\beta}}{se\left(\hat{\hat\beta}\right)}\right\vert$&$\frac{se\left(\hat\beta\right)}{se\left(\hat{\hat\beta}\right)}$&$\left\vert\frac{\hat\beta}{se\left(\hat\beta\right)}\right\vert$&$\left\vert\frac{\hat{\hat\beta}}{se\left(\hat{\hat\beta}\right)}\right\vert$&$\frac{se\left(\hat\beta\right)}{se\left(\hat{\hat\beta}\right)}$&$\left\vert\frac{\hat\beta}{se\left(\hat\beta\right)}\right\vert$&$\left\vert\frac{\hat{\hat\beta}}{se\left(\hat{\hat\beta}\right)}\right\vert$\\
\midrule
$\beta_1$&    1.0003 & 2.0982 & 2.0989 & 1.0021 & 2.1962 & 2.2008 & 1.0011 & 2.2533 & 2.2558 \\      
                                                                                     
$\beta_2$&    0.9998 & 5.4574 & 5.4566 & 1.0013 & 5.4328 & 5.4396 & 1.0016 & 5.3744 & 5.3828 \\      
                                                                                     
$\beta_3$&     0.9999 & 4.7903 & 4.7896 & 1.0013 & 4.8049 & 4.8112 & 1.0017 & 4.7425 & 4.7505 \\      
                                                                                     
$\beta_4$&     1.0001 & 1.9100 & 1.9102 & 1.0017 & 1.9084 & 1.9116 & 1.0022 & 1.9004 & 1.9047 \\      
                                                                                     
$\beta_5$&     0.9998 & 0.2659 & 0.2658 & 1.0025 & 0.5221 & 0.5233 & 1.0038 & 0.5368 & 0.5388 \\      
                                                                                     
$\beta_6$&     1.0005 & 1.5215 & 1.5222 & 1.0029 & 1.4942 & 1.4986 & 1.0028 & 1.4065 & 1.4105 \\
\bottomrule
\end{tabular}
\caption{IV and single Newton-step estimate properties of $\beta_i$ in model (\ref{KKMspec}): standard error ratios and absolute values of t-statistics for $H_0:\beta_i=0$, $i=1,\ldots, 21$. $\beta_1$ is the intercept.}\label{table:empirical_betas}
\end{center}
\end{table}

\clearpage
\begin{subappendices}

\renewcommand{\setthesubsection}{\Alph{subsection}}
 \renewcommand{\thesubsection}{\Alph{subsection}}
\setcounter{lemma}{0}

 \renewcommand{\thelemma}{\Alph{subsection}.\arabic{lemma}}

\setcounter{equation}{0}
\renewcommand\theequation{\Alph{subsection}.\arabic{equation}}

\subsection{Proofs of theorems}\label{Newttheorems}
Write $a_n=p_n+k_n$, $b_n=r_n+k_n$, $c_n=p_nk_n^2+k_n$ and $\tau_n=n^{\frac{1}{2}}/a_n^{\frac{1}{2}}$
.\begin{proof}[Proof of Theorem \ref{approx}] 

\begin{itemize}
\item [(i)] By the mean value theorem (\ref{ivstep}) implies that
\begin{eqnarray}\label{ivappr}
\hat{\hat{\theta}}_{n}-\theta_{0n} & = & \left( I_{a_n}-\hat{H}_n^{-1}\overline{H}_n \right)\left(\hat{\theta}_{n}-\theta_{0n}\right)-\hat{H}_n^{-1}\xi_n\nonumber\\
& = &\hat{\theta}_{n}-\theta_{0n}-\hat{H}_n^{-1}\overline{H}_n \left(\hat{\theta}_{n}-\theta_{0n}\right)-\hat{H}_n^{-1}\xi_n\label{ivbias}
\end{eqnarray}
where $\overline{H}_n={\partial^2\mathcal{Q}_n(\overline{\theta}_{n},\hat{\sigma}_{n}^2)}/{\partial\theta\partial\theta'}$
and $\left\Arrowvert\overline{\theta}_{n}^{}-\theta_{0n}\right\Arrowvert\leq\left\Arrowvert\hat{\theta}_{n}-\theta_{0n}\right\Arrowvert$, with each row of the Hessian matrix evaluated at possibly different $\overline{\theta}_{n}$. The latter point is a technical comment that we take as given in the remainder of the paper whenever a mean-value theorem is applied to vector of values. For any $s{\times}1$ vector $\alpha$, we can use (\ref{ivappr}) to write 
\begin{eqnarray}\label{ivnew1}
\tau_n\alpha'\Psi_n\left(\hat{\hat{\theta}}_{n}-\theta_{0n}\right) & = & \tau_n\alpha'\Psi_n\hat{H}_n^{-1}\left(\hat{H}_n-\overline{H}_n \right)\left(\hat{\theta}_{n}-\theta_{0n}\right) \nonumber \\
& - & \tau_n\alpha'\Psi_n\hat{H}_n^{-1}\xi_n,
\end{eqnarray}
recalling that $\tau_n=n^{\frac{1}{2}}/a_n^{\frac{1}{2}}$. The first term on RHS above has modulus bounded by
$
\tau_n\left\Arrowvert\alpha\right\Arrowvert\left\Arrowvert\Psi_n\right\Arrowvert\left\Arrowvert\hat{H}_n^{-1}\right\Arrowvert\left\Arrowvert\hat{H}_n-\overline{H}_n\right\Arrowvert\left\Arrowvert\hat{\theta}_{n}-\theta_{0n}\right\Arrowvert,
$
where the second factor in norms is $\On\left(a_n^{\frac{1}{2}}\right)$, the third is bounded for sufficiently large $n$ by Lemma \ref{hesseig}, by Lemma \ref{hessian} the fourth is $\Op\left(\max\left\{{p_n^{\frac{3}{2}}b_n^{\frac{1}{2}}}/{n^{\frac{1}{2}}h_n},{b_n^{\frac{1}{2}}c_n^{\frac{1}{2}}}/{n},{p_n^{\frac{1}{2}}b_n^{\frac{1}{2}}}/{n^{\frac{1}{2}}h_n^{\frac{1}{2}}},{b_n}/{n}\right\}\right)$ and the fifth is $\Op\left({b_n^{\frac{1}{2}}}/{n^{\frac{1}{2}}}\right)$ by Theorem 3.1 of \cite{Gupta2013}.  We conclude that the first term on the RHS of (\ref{ivnew1}) is $\Op\left(\max\left\{{p_n^{\frac{3}{2}}}b/{n^{\frac{1}{2}}h_n},b_nc_n^{\frac{1}{2}}/n,p_n^{\frac{1}{2}}b/n^{\frac{1}{2}}h_n^{\frac{1}{2}},b_n^{\frac{3}{2}}/n\right\}\right),$
which is negligible by (\ref{approx2sls1}) and (\ref{approx2sls4}) because
$
\frac{p_n^3b_n^2}{nh_n^2}  \leq  C \left(\frac{p_n^3r_n^2+p_n^3k_n^2}{nh_n^2}\right)=\On\left(\frac{p_n^3}{h_n^2}\frac{r_n^2}{n}+\frac{p_n^3k_n^2}{nh_n^2}\right),
\frac{b_nc_n^{\frac{1}{2}}}{n}  \leq  C \left(\frac{r_np_n^{\frac{1}{2}}k_n+p_n^{\frac{1}{2}}k_n^2}{n^2}\right)
=\On\left(\frac{p_n^{\frac{1}{2}}k_n}{n^{\frac{1}{2}}}\frac{r_n}{n^{\frac{1}{2}}}+\frac{p_n^{\frac{1}{2}}k_n^2}{n^2}\right)$, $\frac{p_nb_n^2}{nh_n}\leq C\left(\frac{p_nr_n^2+p_nk_n^2}{nh_n}\right)=\On\left(\frac{r_n^2}{n}\frac{p_n}{h_n}+\frac{p_nk_n^2}{n}\frac{1}{h_n}\right)$, $\frac{b_n^3}{n^2}\leq C\left(\frac{r_n^3+k_n^3}{n^2}\right)$. For the negligibility of the last term note that $\frac{r_n^{\frac{3}{2}}}{n}=\frac{r_n^2}{n}{r_n^{-\frac{1}{2}}}$.

Thus we only need to find the asymptotic distribution of $-\tau_n\alpha'\Psi_n\hat{H}_n^{-1}\xi_n$. We can write
\begin{equation}\label{ivapprox6}
-\tau_n\alpha'\Psi_n\hat{H}_n^{-1}\xi_n=\frac{2}{\sigma_0^2}\tau_n\alpha'\Psi_n\hat{H}_n^{-1}t_n-\tau_n\alpha'\Psi_n\hat{H}_n^{-1}\phi_n.
\end{equation}

We have 
$
\mathbb{E}\left\Arrowvert\phi_n\right\Arrowvert^2  \leq  \sum_{i=1}^{p_n}\mathbb{E}\left(\ninv  tr C_{in}-{\ninv\sigma_0^{-2}}u'C_{in}u\right)^2 
 =  \sum_{i=1}^{p_n} \mathrm{var} \left(\ninv  u'C_{in}u\right)  = \On\left({p_n}/{nh_n}\right),
$
 (see (A.20) in the proof of Theorem 3.3 and Lemma B.2 in \cite{Gupta2018}), so that 
\begin{equation}\label{phibound}
\left\Arrowvert\phi_n\right\Arrowvert=\Op\left(\frac{p_n^{\frac{1}{2}}}{n^{\frac{1}{2}}h_n^{\frac{1}{2}}}\right).
\end{equation} 
Therefore the second term on the right of (\ref{ivapprox6}) has modulus bounded by $\tau_n$ times
\begin{equation}\label{ivapprox7}
\left\Arrowvert \alpha\right\Arrowvert\left\Arrowvert \Psi_n\right\Arrowvert \;\left\Arrowvert \hat{H}_n^{-1}\right\Arrowvert\;\left\Arrowvert \phi_n\right\Arrowvert,
\end{equation} 
where the second factor is $\On\left(a_n^{\frac{1}{2}}\right)$, the third is bounded for sufficiently large $n$ by Lemma \ref{hesseig} and the last is $\Op\left({p_n^{\frac{1}{2}}}/{n^{\frac{1}{2}}h_n^{\frac{1}{2}}}\right)$. Thus (\ref{ivapprox7}) is $\Op\left({p_n^{\frac{1}{2}}a_n^{\frac{1}{2}}}/{n^{\frac{1}{2}}h_n^{\frac{1}{2}}}\right)$ and the second term on the right of (\ref{ivapprox6}) is $\Op\left({p_n^{\frac{1}{2}}}/{h_n^{\frac{1}{2}}}\right)$ which is negligible by (\ref{approx2sls1}).
Then the asymptotic distribution required is that of
\begin{equation}\label{simpledist}
\frac{2}{\sigma_0^2}\tau_n\alpha'\Psi_n\hat{H}_n^{-1}t_n=\sum_{i=1}^3\Upsilon_{in}+\tau_n\alpha'\Psi_n L_n^{-1}t_n,
\end{equation}

$
\Upsilon_{1n}  =  \frac{2}{\sigma_0^2}\tau_n\alpha'\Psi_n\hat{H}_n^{-1}\left(\hat{H}_n-H_n\right)H_n^{-1}t_n,
\Upsilon_{2n}  =  \frac{2}{\sigma_0^2}\tau_n\alpha'\Psi_n\Xi_n^{-1}\left(H_n-\Xi_n\right)H_n^{-1}t_n,
\Upsilon_{3n}  =  \tau_n\alpha'\Psi_nL_n^{-1}\left(\frac{\sigma_0^2}{2}\Xi_n-L_n\right)\left(\frac{\sigma_0^2}{2}\Xi_n\right)^{-1}t_n.
$
We will demonstrate that $\left\vert \Upsilon_{in}\right\vert=\op(1)$, $i=1,2,3.$
First we observe that
$
\left\vert\Upsilon_{1n}\right\vert\leq \frac{2}{\sigma_0^2}\tau_n\left\Vert\alpha\right\Vert\left\Vert\Psi_n\right\Vert\left\Vert\hat{H}_n^{-1}\right\Vert\left\Vert\hat{H}_n-H_n\right\Vert\left\Vert H_n^{-1}\right\Vert\left\Vert t_n \right\Vert,
$
where the second factor in norms is $\On\left(a_n^{\frac{1}{2}}\right)$, the third and fifth are bounded for sufficiently large $n$ by Lemma \ref{hesseig}, the fourth is $\Op\left(\max\left\{{p_n^{\frac{3}{2}}b_n^{\frac{1}{2}}}/{n^{\frac{1}{2}}h_n},{b_n^{\frac{1}{2}}c_n^{\frac{1}{2}}}/{n},{p_n^{\frac{1}{2}}b_n^{\frac{1}{2}}}/{n^{\frac{1}{2}}h_n^{\frac{1}{2}}},{b_n}/{n}\right\}\right)$ by Lemma \ref{hessian}, and by (A.13) of \cite{Gupta2013} the last is $\Op\left({c_n^{\frac{1}{2}}}/{n^{\frac{1}{2}}}\right)$ . 

Then 
$
\left\vert\Upsilon_{1n}\right\vert=\Op\left(\max\left\{{p_n^{\frac{3}{2}}b_n^{\frac{1}{2}}c_n^{\frac{1}{2}}}/{n^{\frac{1}{2}}h_n},{b_n^{\frac{1}{2}}c_n}/{n},{p_n^{\frac{1}{2}}b_n^{\frac{1}{2}}c_n^{\frac{1}{2}}}/{n^{\frac{1}{2}}h_n^{\frac{1}{2}}},{b_nc_n^{\frac{1}{2}}}/{n}\right\}\right),
$
which is negligible by (\ref{approx2sls1}) and (\ref{approx2sls4}) because
$\frac{p_n^3b_nc_n}{{n}h_n^2}\leq C\left(\frac{p_n^4r_nk_n^2+p_n^4k_n^3}{nh_n^2}\right)=\On\left(\frac{r_n}{n^{\frac{1}{2}}}\frac{p_n^{\frac{3}{2}}k_n^2}{n^{\frac{1}{2}}}\frac{p_n^{\frac{5}{2}}}{h_n^2}+\frac{p_n^3k_n^3}{n}\frac{p_n}{h_n^2}\right)$, $\frac{b_nc_n^2}{n^2}\leq C\left(\frac{r_np_n^2k_n^4+p_n^2k_n^5}{n^2}\right)=\On\left(\frac{r_n}{n^{\frac{1}{2}}}\frac{p_n^2k_n^4}{n^{\frac{3}{2}}}+\frac{p_n^2k_n^4}{n}\frac{k_n}{n}\right)$, $\frac{p_nb_nc_n}{nh_n}\leq C\left(\frac{r_np_n^2k_n^2+p_n^2k_n^3}{nh_n}\right)=$ \\$\On\left(\frac{r_n}{n^{\frac{1}{2}}}\frac{p_n^{\frac{3}{2}}k_n^2}{n^{\frac{1}{2}}}\frac{p_n^{\frac{1}{2}}}{h_n}+\frac{p_n^2k_n^3}{nh_n}\right)$ and $b_nc_n^{\frac{1}{2}}/n=o(1)$ has been shown earlier.

Next $
\left\vert\Upsilon_{2n}\right\vert\leq {2}{\sigma_0^{-2}}\tau_n\left\Vert\alpha\right\Vert\left\Vert\Psi_n\right\Vert\left\Vert H_n^{-1}\right\Vert\left\Vert H_n-\Xi_n\right\Vert\left\Vert \Xi_n^{-1}\right\Vert \left\Vert t_n\right\Vert,
$
where the second factor in norms is $\On\left(a_n^{\frac{1}{2}}\right)$, the third and fifth are bounded for sufficiently large $n$ by Lemma \ref{hesseig}, the fourth is $\Op\left({p_nk_n}/{n^{\frac{1}{2}}}\right)$ by Lemma \ref{hessian2} and the last is $\Op\left({c_n^{\frac{1}{2}}}/{n^{\frac{1}{2}}}\right)$ as above. Then
$
\left\vert\Upsilon_{2n}\right\vert=\Op\left({p_nk_nc_n^{\frac{1}{2}}}/{n^{\frac{1}{2}}}\right)
$
which is negligible by (\ref{approx2sls1}) because ${p_n^2k_n^2c_n}/{n}\leq C{p_n^3k_n^4}/{n}.$
Similarly $\left\vert\Upsilon_{3n}\right\vert=\Op\left({p_nc_n^{\frac{1}{2}}}/{h_n}\right)$ by Lemma \ref{hessian2}, which is negligible by (\ref{approx2sls1}) because ${p_n^2c_n}/{h_n^2}\leq C{p_n^3k_n^2}/{h_n^2}.$ 
Then we only need to find the asymptotic distribution of the last term term in (\ref{simpledist}), but this is precisely the proof of Theorem 3.3 of \cite{Gupta2013}. Replicating those arguments leads to the theorem.
\item [(ii)] In view of Lemmas \ref{hessian2}, \ref{hessianols} and \ref{hesseigols}, the theorem is proved exactly like Theorem \ref{approx} $(i)$, except for different orders of magnitudes of various expressions. In this case two of the orders will be different from the analogous ones considered in the the proof of Theorem \ref{approx} $(i)$. Indeed, the analogue of the bound for the first term in (\ref{ivnew1}) is 
\begin{eqnarray*}
& & \Op\left(n^{\frac{1}{2}}\max\left\{\frac{p_n^{\frac{3}{2}}c_n^{\frac{1}{2}}}{n^{\frac{1}{2}}h_n},\frac{p_n^2}{h_n^2},\frac{p_n^{\frac{7}{4}}c_n^{\frac{1}{4}}}{n^{\frac{1}{4}}h_n^{\frac{3}{2}}},\frac{c_n}{n}\right\}\max\left\{\frac{c_n^{\frac{1}{2}}}{n^{\frac{1}{2}}},\frac{p_n^{\frac{1}{2}}}{h_n}\right\}\right)\\
& = & \Op\left(\max\left\{\pi_{1n},\pi_{2n},\pi_{3n},\pi_{4n},\pi_{5n},\pi_{6n}\right\}\right),
\end{eqnarray*}
where
$
 \pi_{1n}  = {p_n^{\frac{3}{2}}c_n}/{n^{\frac{1}{2}}h_n}$, $\pi_{2n}  = {p_n^2c_n^{\frac{1}{2}}}/{h_n^2}$, $\pi_{3n}  = {p_n^{\frac{7}{4}}c_n^{\frac{3}{4}}}/{n^{\frac{1}{4}}h_n^{\frac{3}{2}}}$,
$ \pi_{4n}  = {c_n^{\frac{3}{2}}}/{n}$, $\pi_{5n} = {n^{\frac{1}{2}}p_n^{\frac{5}{2}}}/{h_n^3}$, $\pi_{6n}  = {n^{\frac{1}{4}}p_n^{\frac{9}{4}}c_n^{\frac{1}{4}}}/{h_n^{\frac{5}{2}}}.
$ Now $\pi_{5n}$ is assumed to tend to zero by (\ref{approxols1}), while  the remaining $\pi_{in}$ terms are also negligible by (\ref{approxols1}) because  $\pi_{1n}^2 =\On\left(p_n^5k_n^4/nh_n^2\right)$, $\pi_{2n}^2 =\On\left(p_n^5k_n^2/h_n^4\right)$, $\pi_{3n}^4 =\On\left(\left(p_n^{4}k_n^3/n\right)\left(p_n^{6}k_n^3/h_n^6\right)\right)$, $\pi_{4n}^2 =\On\left(p_n^3k_n^6/n^2\right)$, $\pi_{6n}^2 =\On\left({n^{\frac{1}{2}}p_n^5k_n}/{h_n^{5}}\right)=$\\$\On\left(\pi_{5n}\left(p_n^{\frac{5}{2}}k_n/h_n^2\right)\right)$. 

 The analogue for the bound on $\Upsilon_{1n}$ is of order 
\begin{eqnarray*}
  &&\On\left(n^{\frac{1}{2}}\right)\Op\left(\max\left\{{p_n^{\frac{3}{2}}c_n^{\frac{1}{2}}}/{n^{\frac{1}{2}}h_n},{p_n^2}/{h_n^2},{p_n^{\frac{7}{4}}c_n^{\frac{1}{4}}}/{n^{\frac{1}{4}}h_n^{\frac{3}{2}}},{c_n}/{n}\right\}\right)\Op\left({c_n^{\frac{1}{2}}}/{n^{\frac{1}{2}}}\right)\\
 &=&  \Op\left(\max\left\{\pi_{1n},\pi_{2n},\pi_{3n},\pi_{4n}\right\}\right),
\end{eqnarray*}
which was shown to be negligible under the assumed conditions. All other bounds remain unchanged and will be also be negligible under (\ref{approxols1}), as in the proof of Theorem \ref{approx} $(i)$.

\end{itemize}
\end{proof}
\begin{proof}[Proof of Theorem \ref{approxbdd}]
Proceeding as in the proof of Theorem \ref{approx} (i), we can write
\begin{eqnarray}
\tau_n\alpha'\Psi_n\left(\hat{\hat{\theta}}_{n}-\theta_{0n}\right) & = & \tau_n\alpha'\Psi_n\hat{H}_n^{-1}\left(\hat{H}_n-\overline{H}_n \right)\left(\hat{\theta}_{n}-\theta_{0n}\right) \nonumber \\
& - & \tau_n\alpha'\Psi_n\left(\hat{H}_n^{-1}-\Xi_n^{-1}\right)\xi_n-\tau_n\alpha'\Psi_n\Xi_n^{-1}\xi_n\label{approxbddpre}.
\end{eqnarray}
 As in the proof of Theorem \ref{approx} (i), the first term on the RHS above is negligible by (\ref{approx2sls3}). Lemma \ref{hesseig} (for bounded $h_n$) indicates that the second term on the RHS of (\ref{approxbddpre}) is bounded in modulus by a constant times
$
\tau_n\left\Vert\Psi_n\right\Vert\left(\left\Vert t_n\right\Vert+\left\Vert\phi_n\right\Vert\right)\left(\left\Vert\hat{H}_n-H_n\right\Vert+\left\Vert{H}_n-\Xi_n\right\Vert\right)
$
which is 
$
\Op\left(n^{\frac{1}{2}}\max\left\{{c_n^{\frac{1}{2}}}/{n^{\frac{1}{2}}},{p_n^{\frac{1}{2}}}/{n^{\frac{1}{2}}h_n^{\frac{1}{2}}}\right\}\max\left\{{p_n^{\frac{3}{2}}b_n^{\frac{1}{2}}}/{n^{\frac{1}{2}}h_n},{b_n^{\frac{1}{2}}c_n^{\frac{1}{2}}}/{n},{p_n^{\frac{1}{2}}b_n^{\frac{1}{2}}}/{n^{\frac{1}{2}}h_n^{\frac{1}{2}}},{b_n}/{n},p_nk_n/n^{\frac{1}{2}}\right\}\right),
$
using (A.13) of \cite{Gupta2013}, (\ref{phibound}) and Lemmas \ref{hessian} and \ref{hessian2} (i). This is negligible by (\ref{approx2sls3}), in a similar way to the preceding proofs. Thus we need to establish the asymptotic distribution of $-\tau_n\alpha'\Psi_n\Xi_n^{-1}\xi_n$,
which is established under the assumed conditions in Theorem 3.4 of \cite{Gupta2018}.
\end{proof}
\subsection{Lemmas}\label{Newtlemmas}

In the subsequent lemmas the assumptions of the theorems that these are used to prove are taken to hold.

\begin{lemma}\label{A'B} (Lemma LS.4 of \cite{Gupta2018}, Supplementary Material) 
$\left\Vert B_n'A_n\right\Vert=\left\Vert A_n'B_n\right\Vert=\Op\left(n^{\frac{1}{2}}p_nk_n\right)$.
\end{lemma}

\begin{lemma}\label{X'B}(Lemma LS.4 of \cite{Gupta2018}, Supplementary Material)
$\left\Vert X_n'B_n\right\Vert=\left\Vert B_n'X_n\right\Vert=\Op\left(n^{\frac{1}{2}}p_n^{\frac{1}{2}}k_n^{\frac{1}{2}}\right)$.
\end{lemma}

\begin{lemma}\label{hessian} 
\mbox{}

$\left\Arrowvert\hat{H}_n-\overline{H}_n\right\Arrowvert$ and $\left\Arrowvert\hat{H}_n-H_n\right\Arrowvert$ are 
$\Op\left(\max\left\{{p_n^{\frac{3}{2}}b_n^{\frac{1}{2}}}/{n^{\frac{1}{2}}h_n},{b_n^{\frac{1}{2}}c_n^{\frac{1}{2}}}/{n},{p_n^{\frac{1}{2}}b_n^{\frac{1}{2}}}/{n^{\frac{1}{2}}h_n^{\frac{1}{2}}},{b_n}/{n}\right\}\right).$
\end{lemma}
\begin{proof}
By the triangle inequality
$
\left\Arrowvert\hat{H}_n-\overline{H}_n\right\Arrowvert \leq \left\Arrowvert\hat{H}_n-H_n\right\Arrowvert+\left\Arrowvert\overline{H}_n^{^{}}-H_n^{}\right\Arrowvert
$, and again by the triangle inequality $\left\Arrowvert\hat{H}_n-H_n\right\Arrowvert$ is bounded by
\begin{equation}
 \frac{2}{n} \left\Vert P_{ji,n}(\hat{\lambda}_{n})-P_{ji,n}\right\Vert  +  \frac{2}{n}\left\vert\frac{1}{\hat{\sigma}_{n}^2}-\frac{1}{\sigma_0^2}\right\vert\left(\left\Vert R_n'R_n\right\Vert + 2\left\Vert X_n'R_n\right\Vert+ \left\Vert X_n'X_n\right\Vert\right) \label{hessdiff}.
\end{equation}
The first term in (\ref{hessdiff}) is bounded by
\begin{equation}\label{hessianpmlex1}
\left\{\sum_{i,j=1}^{p_n} \left(\frac{2}{n} tr (G_{jn}(\hat{\lambda}_{n})G_{in}(\hat{\lambda}_{n})) - \frac{2}{n} tr (G_{jn}G_{in})\right)^2\right\}^{\frac{1}{2}}
\end{equation}
 By the MVT,  we have $tr (G_{jn}(\hat{\lambda}_{n})G_{in}(\hat{\lambda}_{n}))= tr (G_{jn}G_{in})+\overline{\overline{\zeta}}_{ij, n}'\left(\hat{\lambda}_{n}-\lambda_{0n}\right)$, where $\overline{\overline\zeta}_{ij,n}$ has elements $tr\left(G_{in}\left(\overline{\overline{\lambda}}_{n}\right)G_{sn}\left(\overline{\overline{\lambda}}_{n}\right)G_{jn}\left(\overline{\overline{\lambda}}_{n}\right)
+   G_{sn}\left(\overline{\overline{\lambda}}_{n}\right)G_{in}\left(\overline{\overline{\lambda}}_{n}\right)G_{jn}\left(\overline{\overline{\lambda}}_{n}\right)\right)$, $s=1,\ldots,p_n$, and $\left\Arrowvert\overline{\overline\lambda}_{n}^{}-\lambda_{0n}\right\Arrowvert\leq\left\Arrowvert{\hat\lambda}_{n}^{}-\lambda_{0n}\right\Arrowvert$. Thus, the summands in (\ref{hessianpmlex1}) are ${4}{n^{-2}}\;\left(\overline{\overline{\zeta}}_{ij,n}'\left({\hat\lambda}_{n}^{}-\lambda_{0n}\right)\right)^2 \leq  {4}{n^{-2}}\;\left\Vert\overline{\overline{\zeta}}_{ij,n}\right\Vert^2\left\Vert{\hat\lambda}_{n}^{}-\lambda_{0n}\right\Vert^2$
by Cauchy-Schwarz inequality, where the first factor in norms on the RHS is $\On\left({p_n{n^2}/{h_n^2}}\right)$ by Lemma LS.3 of \cite{Gupta2018}, supplementary material. The second factor is bounded by $\left\Arrowvert\hat{\theta}_{n}-\theta_{0n}\right\Arrowvert^2 =  \Op\left(b_n/n\right)$ (see (A.6) of \cite{Gupta2013}),
so we conclude that the summands in (\ref{hessianpmlex1}) are $\Op\left({ b_np_n}/{nh_n^2}\right)$ and therefore (\ref{hessianpmlex1}) is $\Op\left({p_n^{\frac{3}{2}}b_n^{\frac{1}{2}}}/{n^{\frac{1}{2}}h_n}\right)$ and it follows that so is the first term in (\ref{hessdiff}).
By (A.7) of \cite{Gupta2013}, 
\begin{equation}\label{ivsigmainvbound}
\left\vert\frac{1}{\hat{\sigma}_{n}^2}-\frac{1}{\sigma_0^2}\right\vert=\Op\left(\max\left\{\frac{b_n^{\frac{1}{2}}c_n^{\frac{1}{2}}}{n},\frac{p_n^{\frac{1}{2}}b_n^{\frac{1}{2}}}{n^{\frac{1}{2}}h_n^{\frac{1}{2}}},\frac{b_n}{n}\right\}\right),
\end{equation}
which handles the second factor in the second term in (\ref{hessdiff}). We shall now bound the terms inside the parentheses in the second term in (\ref{hessdiff}). These are $\Op(n)$ because $n^{-\frac{1}{2}}\left\Vert R_n\right\Vert=\Op(1)$, $n^{-\frac{1}{2}}\left\Vert X_n\right\Vert=\On(1)$ and $n^{-1}\left\Vert X_n'R_n\right\Vert=\Op(1)$, by Assumption \ref{l2}. 
From (\ref{hessianpmlex1}), (\ref{ivsigmainvbound}), we conclude that
\begin{eqnarray*}
\left\Arrowvert\hat{H}_n-H_n\right\Arrowvert & =& \Op\left(\frac{p_n^{\frac{3}{2}}b_n^{\frac{1}{2}}}{n^{\frac{1}{2}}h_n}\right)  + \Op\left(\max\left\{\frac{b_n^{\frac{1}{2}}c_n^{\frac{1}{2}}}{n},\frac{p_n^{\frac{1}{2}}b_n^{\frac{1}{2}}}{n^{\frac{1}{2}}h_n^{\frac{1}{2}}},\frac{b_n}{n}\right\}\right)
\end{eqnarray*}
Similarly, it may be shown that $\left\Arrowvert\overline{H}_n-H_n\right\Arrowvert$ has the same order, whence the lemma follows.
\end{proof}
\begin{lemma}\label{hessian2}
 (Lemma B.2 of \cite{Gupta2018})
(i) $\left\Arrowvert H_n-\Xi_n\right\Arrowvert=\Op\left({p_nk_n}/{n^{\frac{1}{2}}}\right)$ if  Assumption \ref{errors2} holds, 
(ii) $\left\Arrowvert L_n-\left({\sigma_0^2}/{2}\right)\Xi_n\right\Arrowvert=\On\left({p_n}/{h_n}\right).$
\end{lemma}

\begin{lemma}\label{hesseig} (Lemma B.3 of \cite{Gupta2018}) The following inequalities are satisfied: 
$
\plim\left\Arrowvert\hat{H}_n^{-1}\right\Arrowvert\leq C\plim\left\Arrowvert H_n^{-1}\right\Arrowvert\leq C\lim_{n\rightarrow\infty}\left\Arrowvert\Xi_n^{-1}\right\Arrowvert\leq C\left(\displaystyle\varliminf_{n\rightarrow\infty}\underline{\eta}(L_n)\right)^{-1}\leq C.
$
If $h_n$ does not diverge, the above result becomes 
$
\plim\left\Arrowvert\hat{H}_n^{-1}\right\Arrowvert\leq C\plim\left\Arrowvert H_n^{-1}\right\Arrowvert\leq C\left(\displaystyle\varliminf_{n\rightarrow\infty}\underline{\eta}(\Xi_n)\right)^{-1}\leq C,
$
if also $\displaystyle\varliminf_{n\rightarrow\infty}\underline{\eta}(\Xi_n)>0$.
\end{lemma}

\begin{lemma}\label{hessianols} 
\mbox{}

$\left\Arrowvert\tilde{H}_n-\overline{H}_n\right\Arrowvert$ and $\left\Arrowvert\tilde{H}_n-H_n\right\Arrowvert$ are $\Op\left(\max\left\{{c_n}/{n},{p_n^2}/{h_n^2},{p_n^{\frac{3}{2}}c_n^{\frac{1}{2}}}/{n^{\frac{1}{2}}h_n},{p_n^{\frac{7}{4}}c_n^{\frac{1}{4}}}/{n^{\frac{1}{4}}h_n^{\frac{3}{2}}}\right\}\right).$
\end{lemma}
\begin{proof}
The proof is similar to that of Lemma \ref{hessian} and we only elaborate on the differences from that proof. In this case we need to bound
\begin{equation}
 \frac{2}{n} \left\Vert P_{ji,n}(\tilde{\lambda}_{n})-P_{ji,n}\right\Vert  +  \frac{2}{n}\left\vert\frac{1}{\tilde{\sigma}_{n}^2}-\frac{1}{\sigma_0^2}\right\vert\left(\left\Vert R_n'R_n\right\Vert + 2\left\Vert X_n'R_n\right\Vert+ \left\Vert X_n'X_n\right\Vert\right) \label{hessdiff2}.
\end{equation}
 In the OLS case we have $\tilde{\sigma}_{n}^2-\sigma_0^2=\Op\left(\max\left\{{c_n}/{n},{p_n}/{h_n^2},{p_n^{\frac{1}{2}}c_n^{\frac{1}{2}}}/{n^{\frac{1}{2}}h_n}\right\}\right)$
and $\left\Arrowvert\tilde{\theta}_{n}-\theta_{0n}\right\Arrowvert=\Op\left(\max\left\{{{c_n^{\frac{1}{2}}}/{n^{\frac{1}{2}}},{p_n^{\frac{1}{2}}}/{h_n}}\right\}\right)$,  from (A.23) and (A.21) of \cite{Gupta2013}, respectively. The first term in (\ref{hessdiff2}) is then $\Op\left(\max\left\{{p_n^{\frac{3}{2}}c_n^{\frac{1}{2}}}/{n^{\frac{1}{2}}h_n},{p_n^2}/{h_n^2},{p_n^{\frac{7}{4}}c_n^{\frac{1}{4}}}/{n^{\frac{1}{4}}h_n^{\frac{3}{2}}}\right\}\right)$
while the second one is $\Op\left(\max\left\{{c_n}/{n},{p_n}/{h_n^2},{p_n^{\frac{1}{2}}c_n^{\frac{1}{2}}}/{n^{\frac{1}{2}}h_n}\right\}\right)$.
We may then argue in a similar way that the Hessian evaluated at the OLS estimate differs from its value at an intermediate point in norm by the same to conclude the proof.
\end{proof}

\begin{lemma}\label{hesseigols} (Lemma B.3 of \cite{Gupta2018}) 
$
\plim\left\Arrowvert\tilde{H}_n^{-1}\right\Arrowvert\leq\plim\left\Arrowvert H_n^{-1}\right\Arrowvert\leq\lim_{n\rightarrow\infty}\left\Arrowvert\Xi_n^{-1}\right\Arrowvert\leq\frac{\sigma_0^2}{2}\left(\displaystyle\varliminf_{n\rightarrow\infty}\underline{\eta}(L_n)\right)^{-1}\leq C.
$
\end{lemma}

\end{subappendices}

\clearpage

\bibliographystyle{chicago}
\bibliography{thesisb}

\begin{thebibliography}{}

\bibitem[\protect\citeauthoryear{Andrews}{Andrews}{1997}]{Andrews1997}
Andrews, D. W.~K. (1997).
\newblock A stopping rule for the computation of {Generalized Method of
  Moments} estimators.
\newblock {\em Econometrica\/}~{\em 65}, 913--931.

\bibitem[\protect\citeauthoryear{Anselin}{Anselin}{1988}]{anselin1988spatial}
Anselin, L. (1988).
\newblock {\em {Spatial Econometrics: Methods and Models}}, Volume~4.
\newblock Kluwer Academic Publishers, Boston.

\bibitem[\protect\citeauthoryear{Blommestein}{Blommestein}{1983}]{blommestein1983specification}
Blommestein, H.~J. (1983).
\newblock Specification and estimation of spatial econometric models: $a$
  discussion of alternative strategies for spatial economic modelling.
\newblock {\em Regional Science and Urban Economics\/}~{\em 13}, 251--270.

\bibitem[\protect\citeauthoryear{Case}{Case}{1991}]{case1991spatial}
Case, A.~C. (1991).
\newblock Spatial patterns in household demand.
\newblock {\em Econometrica\/}~{\em 59}, 953--965.

\bibitem[\protect\citeauthoryear{Chudik and Pesaran}{Chudik and
  Pesaran}{2015}]{Chudik2015}
Chudik, A. and M.~H. Pesaran (2015).
\newblock Large panel data models with cross-sectional dependence: {A} survey.
\newblock In B.~H. Baltagi (Ed.), {\em The Oxford Handbook of Panel Data}.
  Oxford University Press.

\bibitem[\protect\citeauthoryear{Cliff and Ord}{Cliff and
  Ord}{1973}]{cliff1973spatial}
Cliff, A.~D. and J.~K. Ord (1973).
\newblock {\em Spatial Autocorrelation}.
\newblock London: Pion.

\bibitem[\protect\citeauthoryear{Conley and Dupor}{Conley and
  Dupor}{2003}]{Conley2003}
Conley, T.~G. and B.~Dupor (2003).
\newblock A spatial analysis of sectoral complementarity.
\newblock {\em Journal of Political Economy\/}~{\em 111}, 311--352.

\bibitem[\protect\citeauthoryear{Das, Kelejian, and Prucha}{Das
  et~al.}{2003}]{das2003}
Das, D., H.~H. Kelejian, and I.~R. Prucha (2003).
\newblock Finite sample properties of estimators of spatial autoregressive
  models with autoregressive disturbances.
\newblock {\em Papers in Regional Science\/}~{\em 82}, 1--26.

\bibitem[\protect\citeauthoryear{Davis}{Davis}{1979}]{Davis1979}
Davis, P.~J. (1979).
\newblock {\em Circulant {M}atrices}.
\newblock Wiley Interscience, New York.

\bibitem[\protect\citeauthoryear{{De Luca}, Magnus, and Peracchi}{{De Luca}
  et~al.}{2018}]{DeLuca2018}
{De Luca}, G., J.~R. Magnus, and F.~Peracchi (2018).
\newblock {Weighted-average least squares estimation of generalized linear
  models}.
\newblock {\em Journal of Econometrics\/}~{\em 204}, 1--17.

\bibitem[\protect\citeauthoryear{Fisher}{Fisher}{1925}]{Fisher1925}
Fisher, R.~A. (1925).
\newblock Theory of statistical estimation.
\newblock {\em Proceedings of the Cambridge Philosophical Society\/}~{\em 22},
  700--725.

\bibitem[\protect\citeauthoryear{Frazier and Renault}{Frazier and
  Renault}{2017}]{Frazier2017}
Frazier, D.~T. and E.~Renault (2017).
\newblock {Efficient two-step estimation via targeting}.
\newblock {\em Journal of Econometrics\/}~{\em 201}, 212--227.

\bibitem[\protect\citeauthoryear{Gupta and Robinson}{Gupta and
  Robinson}{2015}]{Gupta2013}
Gupta, A. and P.~M. Robinson (2015).
\newblock Inference on higher-order spatial autoregressive models with
  increasingly many parameters.
\newblock {\em Journal of Econometrics\/}~{\em 186}, 19--31.

\bibitem[\protect\citeauthoryear{Gupta and Robinson}{Gupta and
  Robinson}{2018}]{Gupta2018}
Gupta, A. and P.~M. Robinson (2018).
\newblock Pseudo maximum likelihood estimation of spatial autoregressive models
  with increasing dimension.
\newblock {\em Journal of Econometrics\/}~{\em 202}, 92--107.

\bibitem[\protect\citeauthoryear{Han, Hsieh, and Lee}{Han
  et~al.}{2017}]{Han2017}
Han, X., C.-s. Hsieh, and L.~F. Lee (2017).
\newblock {Estimation and model selection of higher-order spatial
  autoregressive model: An efficient Bayesian approach}.
\newblock {\em Regional Science and Urban Economics\/}~{\em 63}, 97--120.

\bibitem[\protect\citeauthoryear{Han, Lee, and Xu}{Han et~al.}{2020}]{Han2020}
Han, X., L.~F. Lee, and X.~Xu (2020).
\newblock Large sample properties of {B}ayesian estimation of spatial
  econometric models.
\newblock {\em Econometric Theory Firstview\/}, 39pp.

\bibitem[\protect\citeauthoryear{Hartley and Booker}{Hartley and
  Booker}{1965}]{Hartley1965}
Hartley, H.~O. and A.~Booker (1965).
\newblock {Nonlinear least squares estimation}.
\newblock {\em Annals of Mathematical Statistics\/}~{\em 36}, 638--650.

\bibitem[\protect\citeauthoryear{Helmers and Patnam}{Helmers and
  Patnam}{2014}]{Helmers2014}
Helmers, C. and M.~Patnam (2014).
\newblock Does the rotten child spoil his companion? {S}patial peer effects
  among children in rural {I}ndia.
\newblock {\em Quantitative Economics\/}~{\em 5}, 67--121.

\bibitem[\protect\citeauthoryear{Hsieh and van Kippersluis}{Hsieh and van
  Kippersluis}{2018}]{Hsieh2018}
Hsieh, C.-s. and H.~van Kippersluis (2018).
\newblock {Smoking initiation: Peers and personality}.
\newblock {\em Quantitative Economics\/}~{\em 9}, 825--863.

\bibitem[\protect\citeauthoryear{Hualde and Robinson}{Hualde and
  Robinson}{2011}]{Hualde2011}
Hualde, J. and P.~M. Robinson (2011).
\newblock Gaussian pseudo-maximum likelihood estimation of fractional time
  series models.
\newblock {\em Annals of Statistics\/}~{\em 39}, 3152--3181.

\bibitem[\protect\citeauthoryear{Janssen, Jure\v{c}kova, and
  Veraverbeke}{Janssen et~al.}{1985}]{Janssen1985}
Janssen, P., J.~Jure\v{c}kova, and N.~Veraverbeke (1985).
\newblock {Rate of convergence of one-and two-step M-estimators with
  applications to maximum likelihood and Pitman estimators}.
\newblock {\em The Annals of Statistics\/}~{\em 13}, 1222--1229.

\bibitem[\protect\citeauthoryear{Kasahara and Shimotsu}{Kasahara and
  Shimotsu}{2008}]{Kasahara2008}
Kasahara, H. and K.~Shimotsu (2008).
\newblock {Pseudo-likelihood estimation and bootstrap inference for structural
  discrete Markov decision models}.
\newblock {\em Journal of Econometrics\/}~{\em 146}, 92--106.

\bibitem[\protect\citeauthoryear{Kelejian and Prucha}{Kelejian and
  Prucha}{1998}]{kelejian1998generalized}
Kelejian, H.~H. and I.~R. Prucha (1998).
\newblock A generalized spatial two-stage least squares procedure for
  estimating a spatial autoregressive model with autoregressive disturbances.
\newblock {\em Journal of Real Estate Finance and Economics\/}~{\em 17},
  99--121.

\bibitem[\protect\citeauthoryear{Kelejian and Prucha}{Kelejian and
  Prucha}{1999}]{Kelejian1999}
Kelejian, H.~H. and I.~R. Prucha (1999).
\newblock A generalized moments estimator for the autoregressive parameter in a
  spatial model.
\newblock {\em International Economic Review\/}~{\em 40}, 509--533.

\bibitem[\protect\citeauthoryear{Kolympiris, Kalaitzandonakes, and
  Miller}{Kolympiris et~al.}{2011}]{kolympiris2011spatial}
Kolympiris, C., N.~Kalaitzandonakes, and D.~Miller (2011).
\newblock Spatial collocation and venture capital in the {US} biotechnology
  industry.
\newblock {\em Research Policy\/}~{\em 40}, 1188--1199.

\bibitem[\protect\citeauthoryear{Kristensen and Linton}{Kristensen and
  Linton}{2006}]{Kristensen2006}
Kristensen, D. and O.~Linton (2006).
\newblock {A closed-form estimator for the GARCH(1,1) model}.
\newblock {\em Econometric Theory\/}~{\em 22}, 323--337.

\bibitem[\protect\citeauthoryear{Kristensen and Salani\'{e}}{Kristensen and
  Salani\'{e}}{2017}]{Kristensen2017}
Kristensen, D. and B.~Salani\'{e} (2017).
\newblock {Higher-order properties of approximate estimators}.
\newblock {\em Journal of Econometrics\/}~{\em 198}, 189--208.

\bibitem[\protect\citeauthoryear{Kuersteiner and Prucha}{Kuersteiner and
  Prucha}{2013}]{Kuersteiner2013}
Kuersteiner, G.~M. and I.~R. Prucha (2013).
\newblock {Limit theory for panel data models with cross sectional dependence
  and sequential exogeneity}.
\newblock {\em Journal of Econometrics\/}~{\em 174}, 107--126.

\bibitem[\protect\citeauthoryear{Kuersteiner and Prucha}{Kuersteiner and
  Prucha}{2020}]{Kuersteiner2020}
Kuersteiner, G.~M. and I.~R. Prucha (2020).
\newblock {Dynamic spatial panel models: networks, common shocks, and
  sequential exogeneity}.
\newblock {\em Econometrica\/}~{\em 88}, 2109--2146.

\bibitem[\protect\citeauthoryear{Kwok}{Kwok}{2019}]{Kwok2019}
Kwok, H.~H. (2019).
\newblock {Identification and estimation of linear social interaction models}.
\newblock {\em Journal of Econometrics\/}~{\em 210}, 434--458.

\bibitem[\protect\citeauthoryear{Lahiri}{Lahiri}{1996}]{Lahiri1996}
Lahiri, S.~N. (1996).
\newblock On inconsistency of estimators based on spatial data under infill
  asymptotics.
\newblock {\em Sankhy\={a}: Series A\/}~{\em 58}, 403--417.

\bibitem[\protect\citeauthoryear{LeCam}{LeCam}{1956}]{LeCam1956}
LeCam, L. (1956).
\newblock On the asymptotic theory of estimation and testing hypotheses.
\newblock {\em Proceedings of the Third Berkeley Symposium on Mathematical
  Statistics and Probability\/}~{\em 1}, 129--156.

\bibitem[\protect\citeauthoryear{Lee}{Lee}{2002}]{lee2002consistency}
Lee, L.~F. (2002).
\newblock Consistency and efficiency of least squares estimation for mixed
  regressive, spatial autoregressive models.
\newblock {\em Econometric Theory\/}~{\em 18}, 252--277.

\bibitem[\protect\citeauthoryear{Lee}{Lee}{2004}]{lee2004asymptotic}
Lee, L.~F. (2004).
\newblock Asymptotic distributions of quasi-maximum likelihood estimators for
  spatial autoregressive models.
\newblock {\em Econometrica\/}~{\em 72}, 1899--1925.

\bibitem[\protect\citeauthoryear{Lee}{Lee}{2007}]{Lee2007}
Lee, L.~F. (2007).
\newblock {GMM and 2SLS} estimation of mixed regressive, spatial autoregressive
  models.
\newblock {\em Journal of Econometrics\/}~{\em 137}, 489--514.

\bibitem[\protect\citeauthoryear{Lee and Liu}{Lee and Liu}{2010}]{Lee2010}
Lee, L.~F. and X.~Liu (2010).
\newblock Efficient {G}{M}{M} estimation of high order spatial autoregressive
  models with autoregressive disturbances.
\newblock {\em Econometric Theory\/}~{\em 26}, 187--230.

\bibitem[\protect\citeauthoryear{Lee and Yu}{Lee and Yu}{2013}]{Lee2013}
Lee, L.~F. and J.~Yu (2013).
\newblock Near unit root in the spatial autoregressive model.
\newblock {\em Spatial Economic Analysis\/}~{\em 8}, 314--351.

\bibitem[\protect\citeauthoryear{Li}{Li}{2017}]{Li2017}
Li, K. (2017).
\newblock {Fixed-effects dynamic spatial panel data models and impulse}.
\newblock {\em Journal of Econometrics\/}~{\em 198}, 102--121.

\bibitem[\protect\citeauthoryear{Lin and Lee}{Lin and Lee}{2010}]{Lin2010}
Lin, X. and L.~F. Lee (2010).
\newblock {G}{M}{M} estimation of spatial autoregressive models with unknown
  heteroskedasticity.
\newblock {\em Journal of Econometrics\/}~{\em 157}, 34--52.

\bibitem[\protect\citeauthoryear{Liu and Yang}{Liu and Yang}{2015}]{Liu2015}
Liu, S.~F. and Z.~Yang (2015).
\newblock Modified {QML} estimation of spatial autoregressive models with
  unknown heteroskedasticity and nonnormality.
\newblock {\em Regional Science and Urban Economics\/}~{\em 52}, 50--70.

\bibitem[\protect\citeauthoryear{Ord}{Ord}{1975}]{Ord1975}
Ord, K. (1975).
\newblock {Estimation methods for models of spatial interaction}.
\newblock {\em Journal of the American Statistical Association\/}~{\em 70},
  120--126.

\bibitem[\protect\citeauthoryear{Ortega and Rheinboldt}{Ortega and
  Rheinboldt}{1970}]{Ortega1970}
Ortega, J. and W.~Rheinboldt (1970).
\newblock {\em {Iterative Solution of Nonlinear Equations in Several
  Variables}}.
\newblock Academic Press, New York.

\bibitem[\protect\citeauthoryear{Pace and Barry}{Pace and
  Barry}{1997}]{Pace1997}
Pace, R.~K. and R.~Barry (1997).
\newblock Quick computation of spatial autoregressive estimators.
\newblock {\em Geographical Analysis\/}~{\em 29}, 232--247.

\bibitem[\protect\citeauthoryear{Pinkse, Slade, and Brett}{Pinkse
  et~al.}{2002}]{pinkse2002}
Pinkse, J., M.~E. Slade, and C.~Brett (2002).
\newblock Spatial price competition: A semiparametric approach.
\newblock {\em Econometrica\/}~{\em 70}, 1111--1153.

\bibitem[\protect\citeauthoryear{Robinson}{Robinson}{1988}]{Robinson1988}
Robinson, P.~M. (1988).
\newblock The stochastic difference between econometric statistics.
\newblock {\em Econometrica\/}~{\em 56}, 531--548.

\bibitem[\protect\citeauthoryear{Robinson}{Robinson}{2005}]{Robinson2005a}
Robinson, P.~M. (2005).
\newblock {Efficiency improvements in inference on stationary and nonstationary
  fractional time series}.
\newblock {\em Annals of Statistics\/}~{\em 33}, 1800--1842.

\bibitem[\protect\citeauthoryear{Robinson}{Robinson}{2010}]{robinson2010efficient}
Robinson, P.~M. (2010).
\newblock Efficient estimation of the semiparametric spatial autoregressive
  model.
\newblock {\em Journal of Econometrics\/}~{\em 157}, 6--17.

\bibitem[\protect\citeauthoryear{Robinson}{Robinson}{2011}]{Robinson2011}
Robinson, P.~M. (2011).
\newblock {Asymptotic theory for nonparametric regression with spatial data}.
\newblock {\em Journal of Econometrics\/}~{\em 165}, 5--19.

\bibitem[\protect\citeauthoryear{Rothenberg}{Rothenberg}{1984}]{Rothenberg1984}
Rothenberg, T.~J. (1984).
\newblock Approximate normality of generalized least squares estimates.
\newblock {\em Econometrica\/}~{\em 52}, 811--825.

\bibitem[\protect\citeauthoryear{Rothenberg and Leenders}{Rothenberg and
  Leenders}{1964}]{Rothenberg1964}
Rothenberg, T.~J. and C.~T. Leenders (1964).
\newblock Efficient estimation of simultaneous equation systems.
\newblock {\em Econometrica\/}~{\em 32}, 57--76.

\bibitem[\protect\citeauthoryear{Zhu, Huang, Pan, and Wang}{Zhu
  et~al.}{2020}]{Zhu2020}
Zhu, X., D.~Huang, R.~Pan, and H.~Wang (2020).
\newblock Multivariate spatial autoregressive model for large scale social
  networks.
\newblock {\em Journal of Econometrics\/}~{\em 215}, 591--606.

\end{thebibliography}
\end{document}